\documentclass[reqno]{amsart}

\usepackage{amssymb}
\usepackage{enumerate}
\usepackage{graphicx}

\usepackage{hyperref}

\numberwithin{equation}{section}

\newcommand{\nul}{\mathbf{0}}

\newcommand{\qtx}[1]{\quad\text{#1}\quad}

\newcommand{\beq}{\begin{equation}}
\newcommand{\eeq}{\end{equation}}

\newcommand{\bl}{\begin{lemma}}
\newcommand{\el}{\end{lemma}}
\newcommand{\bc}{\begin{coro}}
\newcommand{\ec}{\end{coro}}
\newcommand{\bp}{\begin{prop}}
\newcommand{\ep}{\end{prop}}
\newcommand{\bd}{\begin{defini}}
\newcommand{\ed}{\end{defini}}

\DeclareMathOperator{\im}{\mathrm{Im}}  
\DeclareMathOperator{\re}{\mathrm{Re}}

\newtheorem{theorem}{Theorem}[section]
\newtheorem{lemma}[theorem]{Lemma}
\newtheorem{coro}[theorem]{Corollary}
\newtheorem{prop}[theorem]{Proposition}
\newtheorem{remark}[theorem]{Remark} 
\newtheorem{defini}[theorem]{Definition}

\newtheorem*{assumptions}{Assumptions}

\newcommand{\lb}{\lambda}
\newcommand{\ze}{\zeta}

\newcommand{\bvp}{{\boldsymbol{\varphi}}}

\newcommand{\Dd}{\mathcal{D}}

\newcommand{\I}{\mathcal{I}}
\newcommand{\Hh}{\mathcal{H}}
\newcommand{\Ff}{\mathcal{F}}
\newcommand{\Kk}{\mathcal{K}}
\newcommand{\Tt}{\mathcal{T}}
\newcommand{\cB}{\mathcal{B}}
\newcommand{\Bb}{\mathcal{B}}
\newcommand{\Aa}{\mathcal{A}}
\newcommand{\Mm}{\mathcal{M}}
\newcommand{\Nn}{\mathcal{N}}

\newcommand{\Ll}{\mathcal{L}}
\newcommand{\Tr}{{\rm Tr}}
\newcommand{\Cc}{\mathcal{C}}

\newcommand{\Pp}{\mathcal{P}}
\newcommand{\Ee}{\mathcal{E}}
\newcommand{\Ss}{\mathcal{S}}
\newcommand{\Oo}{\mathcal{O}}

\newcommand{\RR}{\mathbb{R}}
\newcommand{\NN}{\mathbb{N}}
\newcommand{\ZZ}{\mathbb{Z}}

\newcommand{\CC}{\mathbb{C}}
\newcommand{\VV}{\mathbb{V}}
\newcommand{\TT}{\mathbb{T}}

\newcommand{\E}{\mathbb E}
\newcommand{\G}{\mathbb G}

\newcommand{\aaa}{\mathbf{a}}
\newcommand{\bbb}{\mathbf{b}}

\newcommand{\Sym}{{\rm Sym}}
\newcommand{\diag}{{\rm diag}}
\newcommand{\sgn}{{\rm sgn}}

\newcommand{\hn}{|\!|\!|}

\newcommand{\hnn}{|\!|\!|\!|}

\newcommand{\PE}{\mbox{$\Pp\Ee$}}

\newcommand{\vze}{\vec{\zeta}}
\newcommand{\vxi}{\vec{\xi}}

\newcommand{\one}{{\bf 1}}
\newcommand{\blb}{{\boldsymbol\theta}}

\newcommand{\PP}{\mathfrak{P}}

\newcommand{\bpart}{\boldsymbol{\partial}}

\newcommand{\brc}{\mathbf{c}}

\newcommand{\III}{\boldsymbol{\I}}

\newcommand{\pmat}[1]{\begin{pmatrix} #1  \end{pmatrix}}
\newcommand{\smat}[1]{\left( \begin{smallmatrix} #1  \end{smallmatrix} \right)}

\newcommand{\wt}{\widetilde}

\begin{document}

\title[A.C. Spectrum on the Fibonacci and similar tree-strips]
{Absolutely Continuous Spectrum for random Schr\"odinger operators on the Fibonacci and similar tree-strips}

\author{Christian Sadel}
\address{Mathematics Department, University of British Columbia, Vancouver, BC, V6T 1Z2,  Canada; and Institute of Science and Technology Austria, 3400 Klosterneuburg}
\email{Christian.Sadel@ist.ac.at}

\thanks{This research was supported by NSERC Discovery grant 92997-2010 RGPIN and 
by the People Programme (Marie Curie Actions) of the EU 7th Framework Programme FP7/2007-2013, REA grant 291734.}



\begin{abstract}
We will consider cross products of finite graphs with a class of trees that have arbitrarily but finitely long line segments, such as the
Fibonacci tree. Such cross products are called tree-strips.
We prove that for small disorder random Schr\"odinger operators on such tree-strips have purely absolutely continuous spectrum in
a certain set.
\end{abstract}

\maketitle 

\section{Introduction} 

It will be most convenient to describe the trees considered in this work by a substitution rule given by a substitution matrix 
$(S_{pq})_{p,q=0}^L \in \ZZ_+^{(L+1)\times(L+1)}$ with positive integer entries.
The trees are constructed starting from a root and the matrix gives the rule how to substitute a vertex by its children going to the next generation. 
Here, 'generation' describes the graph distance from the root, the 'children' of a vertex
are all connected neighbors whose graph distance to the root is increased by one, the other neighbor will be called 'parent'.

The precise substitution rule is the following: Each vertex $x$ of the tree has a label $l(x)\in\{0,\ldots,L\}$, one starts with the root with a certain label, then each vertex of label $p$ has exactly $S_{pq}$ children of label $q$.
The (isomorphy class of the) tree is then determined by the matrix $S$ and the label of the root.

We will consider the trees for the substitution matrices $S^{K,L}=[S^{K,L}_{p,q}]_{p,q=0}^L$ with  
$S^{K,L}_{0,0}=K$, $S^{K,L}_{p,p+1}=1$, $S^{K,L}_{L,0}=1$, and all other entries $0$,
where $K\geq 1$ and $L\geq 1$ are integers,
\begin{equation}\label{eq-def-S_L,K}
S^{K,L}\,=\,\pmat{K & 1 & \\ & & \ddots \\ & & & 1 \\ 1 & 0 & \cdots & 0}\;, \qtx{e.g. for $L=1$ ,}
S^{K,1}=\pmat{K & 1 \\ 1 & 0}\,.
\end{equation}
 We denote the tree with root label $p\in\{0,\ldots,L\}$ and substitution matrix $S^{K,L}$ by $\TT^{K,L}_p$ and then define the forest $\TT^{K,L}$ to be the disjoint, disconnected union of the $\TT^{K,L}_p$, i.e. $\TT^{K,L}=\bigcup_{p=0}^L \TT^{K,L}_p$.
Another way to think of the tree $\TT^{K,L}_0$ is to start with the rooted Bethe lattice
where each vertex has $K+1$ children (i.e. the root has $K+1$ neighbors and any other vertex has $K+2$ neighbors) 
and then for each vertex one takes one of the $K+1$ forward edges (going to the next generation) and puts $L$ additional vertices on them (cf. Figure~\ref{fig0}).

\begin{figure}[ht]
\begin{center}
 \includegraphics[width=9cm]{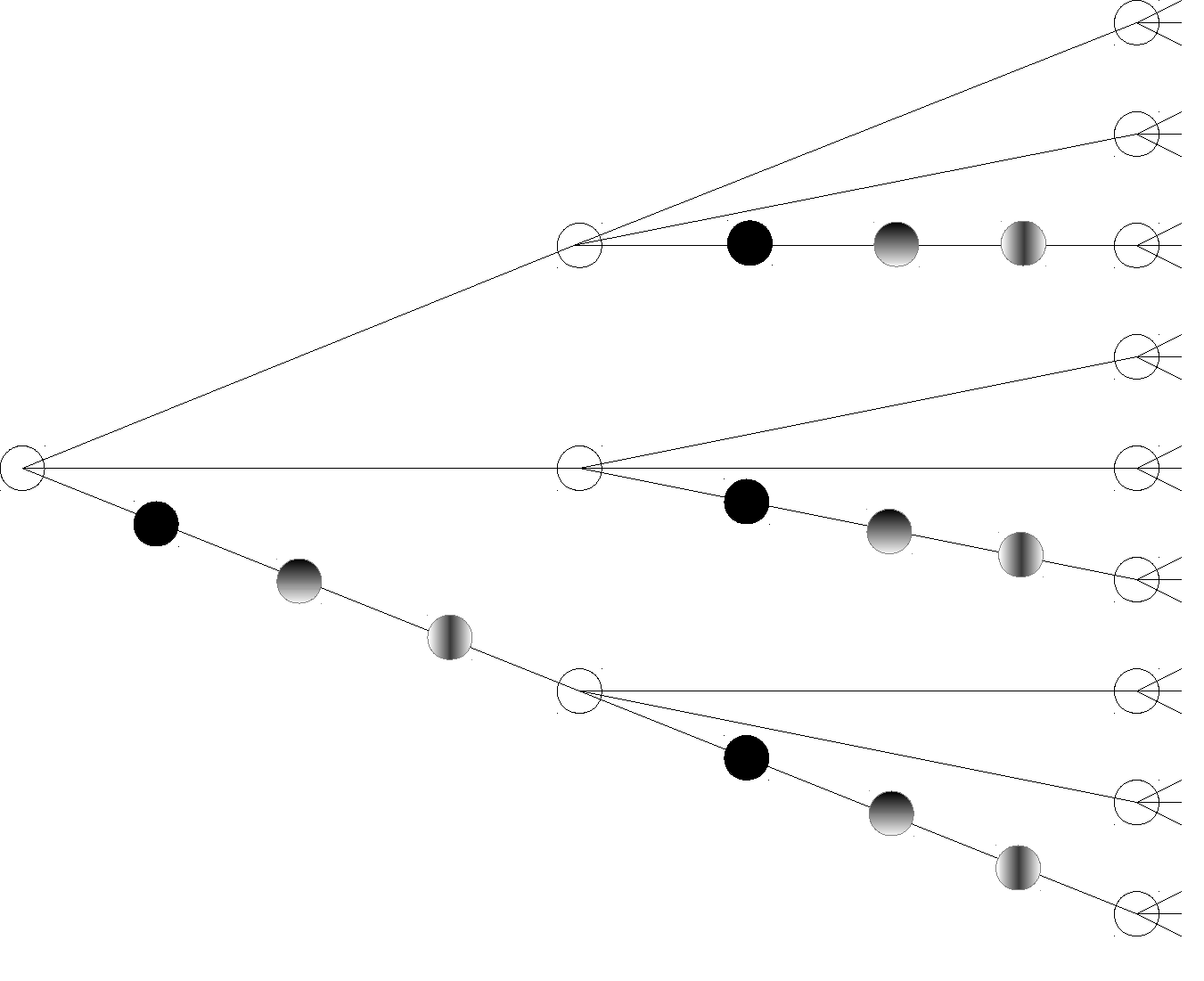}
\end{center}
\caption{The tree $\TT^{2,3}_0$, $K=2,$ $L=3$.
The open circles are vertices of label $0$, the full filled circles are vertices of label $1$, each vertex of label $1$ is followed by one vertex of label $2$ and one of label $3$. These labels are indicated by different shadings of the circles.}\label{fig0}
\end{figure}

The Fibonacci trees are the trees associated to the substitution matrix 
$S^{1,1} =\smat{ 0 & 1 \\ 1 & 1 }$, i.e. each vertex $x$ of the tree has either the label $l(x)=0$ or $l(x)=1$, each vertex with label $0$ has a child with label $0$ and one with label $1$ and each child with label $1$ has one child with label $0$.
So each vertex of label $1$, except for possibly the root, has 2 neighbors (one parent and 1 child), and each vertex of label $0$, except
for possibly the root, has 3 neighbors (one parent and 2 children).
(see Figure~\ref{fig1}).

These trees are called Fibonacci trees for the following reason.
Let $\#_n(\TT^{1,1}_{j})$ denote the number of vertices in the $n$-th generation of the tree $\TT^{1,1}_j$, the root being the first generation.
Moreover, let $f_n$ denote the $n$-th Fibonacci number starting with $f_1=1,\;f_2=1$. Then, one has
$\#_n(\TT^{1,1}_{1})=f_n$ and $\#_n(\TT^{1,1}_{0})=f_{n+1}$.

\begin{figure}[ht]
\begin{center}
 \includegraphics[width=7cm]{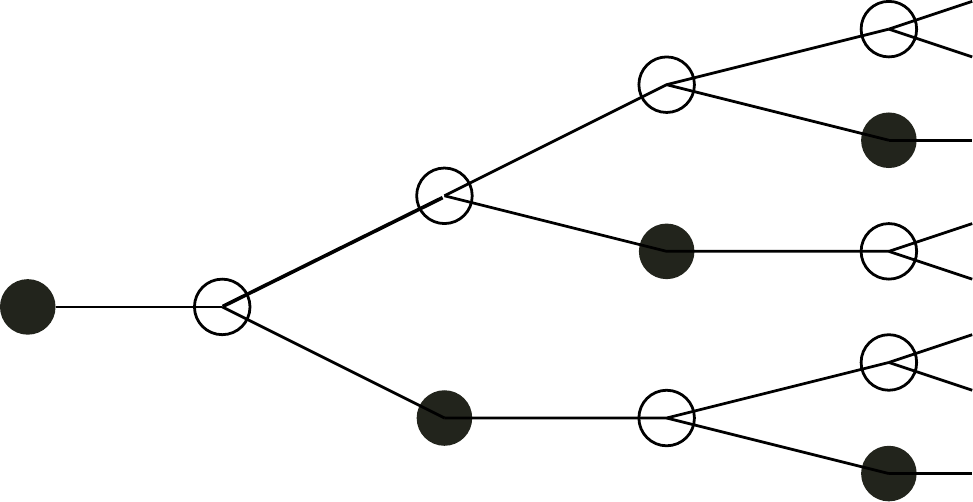}
\end{center}
\caption{Fibonacci tree $\TT^{1,1}_1$. Vertices of label 1 are filled circles and vertices of label 0 are non-filled circles.
The trees $\TT^{1,L}_1$ look similar where each filled circle has to be replaced by
a chain (line-segment) of $L$ vertices.
}\label{fig1}
\end{figure}

By $d(x,y)$ we denote the graph distance of $x,y \in \TT^{K,L}$, where
$d(x,y)=\infty$ if $x$ and $y$ are elements of different connected components $\TT^{K,L}_p$.

A tree-strip is the cross product of a tree with a finite set $\I=\{1,\ldots,m\}$.
On $\ell^2(\TT^{K,L}\times\I)\cong \ell^2(\TT^{K,L},\CC^m)=\{u:\TT^{K,L}\to\CC^m,\,\sum_{x\in\TT^{K,L}} \|u(x)\|^2 < \infty\}$ 
which is also canonically equivalent to 
$\bigoplus_{p=0}^L \ell^2(\TT^{K,L}_p,\CC^m)$ and
$\ell^2(\TT^{K,L})\otimes\CC^m$, we define the random operators
\begin{equation} \label{eq-H_lb}
(H_{\lb} u)(x)=\left(\sum_{y:d(x,y)=1} u(y)\right) + Au(x) + \lb V(x) u(x)\,.
\end{equation}
Here, $A\in \Sym(m)$ represents the 'free vertical operator'
and the matrices $V(x)\in\Sym(m)$ for $x\in\TT^{K,L}$ are independent identically distributed random variables, 
distributed according to some probability measure $\nu$ on $\Sym(m)$ and scaled by the coupling constant $\lambda$.
$\Sym(m)$ denotes the set of real symmetric $m\times m$ matrices.
These operators might be either thought of to model one particle on the product $\TT^{K,L}\times\I$
or to model one particle on $\TT^{K,L}$
with internal degrees of freedom and random hopping between these internal degrees, described by $A$ and $V(x)$.
Clearly, $H_\lb=\bigoplus_{p=0}^L H^{(p)}_\lambda$, where $H^{(p)}_{\lb}$ is the restriction of $H_\lb$ to $\ell^2(\TT^{K,L}_p,\CC^m)$ and can be seen
as  random Schr\"odinger operator on the tree strip $\TT^{K,L}_p\times \I$.

If $\I=\G$ is a finite graph, then $\TT^{K,L} \times \G$ can be interpreted as the product graph where
$(x,k), (y,j) \in\TT^{K,L} \times \G$ are connected by an edge, if either $x=y$ and $k, j\in\G$ are connected by an edge,
or $k=j$ and $x,y \in \TT^{K,L}$ are connected by an edge. 
If $A$ is chosen to be the adjacency matrix of $\G$ and $V(x)$ diagonal with i.i.d. entries, then $H_0$ is the adjacency operator on this product graph and $H_\lb$ corresponds to the Anderson model on this product graph.
 
For the Anderson model on $\ZZ^d$ or $\RR^d$ in any dimension $d$, Anderson localization is proved at spectral edges and for high disorder
\cite{FS,FMSS,DLS,SW,CKM,DK,Kl2,AM,A,Wang,Klo}. It is also known to hold for one dimensional \cite{GMP,KuS,CKM}
and quasi-one dimensional models like strips \cite{Lac,KLS}
and finite dimensional trees \cite{Breu}, except if a built in symmetry prevents localization as e.g. in \cite{SS}.
In dimensions $d\geq 3$ the Anderson model is expected to have some absolutely continuous spectrum (short a.c. spectrum) for low disorder whereas for $d=2$ one expects localization. These conjectures remain open problems.
The existence of a.c. spectrum has only been proved for the Anderson model on trees and other
tree-like graphs of infinite dimension with exponentially growing boundary
\cite{Kl3, Kl4, Kl6, ASW, FHS, FHS2, H, AW, KLW, KLW2, FHH, KS, Sad, Sha}.
This work adds some more examples to this list.
It appears that the hyperbolic nature of such graphs leads to conservation of
a.c. spectrum and ballistic dynamical behavior \cite{Kl5, KS2, AW2} and these results
should hold for much more general hyperbolic graphs. Therefore, it may be worth it to further
generalize the results and identify the technical problems occurring in this process.
Also, a recent review emphasized the importance of trees of finite cone type
\cite{KLW3} and the trees $\TT^{K,L}_p$ belong to this class.
If a tree has an assigned root $0$ then the $n$-th generation is the set of vertices $x$ with graph distance $d(x,0)=n$.
The cone of descendants of $x$ is then defined as the set of vertices $y$, where the shortest path to the root goes through $x$, i.e. 
the set of $y$ such that $d(0,y)=d(0,x)+d(x,y)$.
The phrase 'finite cone type' refers to the fact, that there are only finitely many different (isomorphy classes of) 
cones of descendants. Using the isomorphy class of the cone as label, each tree of finite cone type can be associated to a substitution matrix, and each tree associated to a substitution matrix is a tree of finite cone type.
For the case considered here, the trees $\TT^{K,L}_p$ for $p\in\{0,\ldots,L\}$ 
are exactly the different isomorphy classes of cones of descendants.

In my previous work \cite{Sad} I already considered random Schr\"odinger operators on tree-strips of finite cone type. However, none of these trees $\TT^{K,L}_p$ were covered there.
One of the main assumptions needed in \cite{Sad} was that every vertex has at least 2 children which played a significant role at various places. 
This means the trees could not have any line segment, that is a vertex or chain of vertices not being the root which has only two neighbors, one parent and one child.
The trees $\TT_p^{K,L}$ have line segments of length $L$. In some sense these are the simplest trees with that property. 
The main argument in \cite{Sad} is adapted from \cite{Kl3, KS} and uses a fixed point equation and
the Implicit Function Theorem performed in some Banach spaces
that are associated to supersymmetric functions.
As we will see, for the tree-strips considered here this technique still works, but there are quite a few technical subtleties 
that are pointed out in this work. 
The Implicit Function Theorem has to be applied in a slightly different Banach space
($\Hh_\infty \times \Hh^L$ instead of $\Hh_\infty^{L+1}$).
In general, $\Hh_p$, $1\leq p\leq \infty$ is in some sense the intersection of a supersymmetric $L^2$ and $L^\infty$ space.
For the set up of the fixed point equation in \cite{Sad} it was important to always have a product of at least two such functions which then is a $L^2$ and $L^1$ function, so that the Fourier transform is mapping them back to a $L^2$ and $L^\infty$ function. The line segments in $\TT^{K,L}_p$ lead to the fact that we do not have such products of functions so that a Fourier transform is just giving an $L^2$ but not necessarily an $L^\infty$ function.
The change of this Banach spaces also leads to adjustments in the inductive Proposition~\ref{prop-xize} and the final arguments giving a continuous extension of the Green's matrix to real energies which is given by certain integrals (cf. \eqref{eq-EG-x-D} and \eqref{eq-EGG-x-D}).
For this it is important that the terms inside the integral extend continuously in a supersymmetric $L^1$, something that can still be achieved here.
In former works \cite{KS, Sad} these terms even extended in $\Hh_1$, which is not anymore the case here.
For the analysis of the Frechet derivative, compactness of a certain operator is needed which 
also demands some additional work in this case (cf. Proposition~\ref{lem-op-compact}).
It relies on the identities mentioned in Appendix~\ref{app-FMF}.
Some of the different used arguments need a stronger assumption on the distribution of the
matrix valued potential $V(x)$, namely it has to be compactly supported.
Another new aspect in this work is the use of the identity given in Proposition~\ref{prop-gr} to obtain that a certain Frechet derivative is invertible.
This method would not work for all the trees considered in \cite{Sad} where an extremely technical described set of energies where the Frechet derivative
is not invertible, had to be removed.

For the Fibonacci tree one can explicitly calculate the spectrum for the
adjacency operator (cf. Proposition~\ref{prop-spectrum}), therefore the main theorem is less technical for this case.

Moreover, in \cite[Theorem 1.2]{Sad}  some
set of energies had to be excluded to get the almost sure a.c. spectrum. 
This set was given by a very technical condition which was shown to remove a nowhere
dense set in a certain case.
In Lemma~\ref{lem-not-1} 
we show that this condition is never satisfied for the trees considered here,
hence we do not have to remove certain energies.
The argument is based on an identity satisfied by the Green's functions and shown in Appendix~\ref{app-gamma}.

\vspace{.2cm}

Considering \eqref{eq-H_lb},
let us remark that there is an orthogonal matrix $O\in{\rm O}(m)$ such that $O^\top A O$ is diagonal.
Then $(\one\otimes O)$ is unitary and one obtains the equivalent family of operators
\begin{equation}
 \left[ (\one\otimes O)^* H_\lambda (\one\otimes O) \right] u(x)=
\left(\sum_{y:d(x,y)=1} u(y)\right) + O^\top A O u(x) + \lambda O^\top V(x) O u(x)\;.\notag
\end{equation}
Hence, without loss of generality, we can assume that 
$A$ is a diagonal matrix and we will do so in the proofs.
In particular, the non-random operator $H_0$ is unitarily equivalent to a direct sum of shifted adjacency operators on $\TT^{K,L}$,
$H_0=\Delta \otimes \one + \one \otimes A \cong \bigoplus_{j=1}^m \Delta+a_j$ 
on $\ell^2(\TT^{K,L})\otimes\CC^m$ where the $a_j$ are the eigenvalues of $A$ and $\Delta$
describes the adjacency operator on $\ell^2(\TT^{K,L})$ given by
\begin{equation}
 (\Delta v)(x)=\sum_{y:d(x,y)=1} v(y)\;,\qquad v\in\ell^2(\TT^{K,L})\;.
\end{equation}

Our interest lies in the spectral type of $H_\lb$.
In order to state the main theorems we have to consider the adjacency operator first.
From now on we have a fixed $K,L$ and will omit these indices in many future defined quantities. The root of $\TT^{K,L}_p$ will be called $0^{(p)}$. For $x\in\TT^{K,L}$ we let $|x\rangle\in\ell^2(\TT^{K,L})$ denote the element given by $|x\rangle(y)=\delta_{x,y}=\begin{cases} 1 & :\;x=y \\ 0 & :\;x\neq y \end{cases}$. 
For some operator $H$, $\langle x| H|y\rangle$ denotes the scalar product between $|x\rangle$ and $H|y\rangle$, where we use the physics convention that the scalar product is anti-linear in the first and linear in the second component.
For $p\in\{0,\ldots,L\}$ and $\im(z)>0$ we define the Green's functions
\begin{equation}\label{eq-def-Gamma}
 \Gamma^{(p)}_z:= \;\langle 0^{(p)}|(\Delta-z)^{-1}|0^{(p)}\rangle\;.
\end{equation}
For $E\in \RR$ we further define $\Gamma^{(p)}_E$ by the limit
\begin{equation}\label{eq-def-Gamma-E}
 \Gamma^{(p)}_E:=\,\lim_{\eta\downarrow 0} \Gamma^{(p)}_{E+i\eta}\quad\text{if the limit exists in $\CC$.} 
\end{equation}
Let us define the following sets of energies $E\in\RR$, 
\begin{equation}\label{eq-def-I_p^K,L}
 I_p^{K,L}:=\{E\in \RR\,:\,\Gamma^{(p)}_E\;\text{exists, and}\;
 \im(\Gamma^{(p)}_E) > 0\;\}\,.
\end{equation}
Furthermore, for $d\in \NN$ let $\Delta_d$ denote the $d\times d$
adjacency matrix for the finite line with $d$ vertices and $d-1$ edges, i.e.
\begin{equation}\label{eq-def-Delta-d}
 \Delta_d=\pmat{0 & 1 \\ 1 & \ddots & \ddots \\ & \ddots & \ddots & 1 \\ & & 1 & 0} \in 
 {\rm Mat}(d\times d,\RR)
 \qtx{with the convention} \Delta_1 = 0\;.
\end{equation}
 We set
\begin{equation}\label{eq-def-EL}
\Ee_L:=\bigcup_{d=1}^L \sigma(\Delta_d) \qtx{and}
I^{K,L}:= I_0^{K,L} \;\setminus \; \Ee_L
\end{equation}
where $\sigma(\Delta_d)$ denotes the set of eigenvalues of the matrix $\Delta_d$ and therefore,
$\Ee_L$ is a finite set.
\begin{prop} \label{prop-spectrum} We have:
$ $ \begin{enumerate}[{\rm (i)}]
 \item For all $p\in\{0,\ldots,L\}$,  $I_p^{K,L}=I^{K,L}_0$.  
 $I^{K,L}_0$ and $I^{K,L}$ are non-empty unions of finitely many open intervals
 and for all $p$, $\Gamma^{(p)}_E$ depends analytically on $E\in I_0^{K,L}$.
 \item  The absolutely continuous spectrum of the adjacency operator $\Delta$ on $\ell^2(\TT^{K,L})$ is given by the closure
 $$
 \sigma_{ac}(\Delta)\,=\,\overline{I^{K,L}}\,=\,\overline{I^{K,L}_0}\;.
 $$
 \item For any fixed $L$ and $E_0$ there is a $K_0$ such that for $K>K_0$ the closure of
 $I^{K,L}$ includes the interval $[-E_0,E_0]$, i.e.
 \begin{equation}
  [-E_0\,,\,E_0]\,\subset\, \overline{I^{K,L}}\qtx{for all} K>K_0=K_0(L,E_0)\;.
 \end{equation}

 \item For the Fibonacci trees ($K=L=1$) we have
 \begin{equation} 
 I^{1,1}=\left(-\frac32 \sqrt{3}\,,\,0\right)\,\cup\,\left(0\,,\,\frac32\sqrt{3} \right)\;.
 \end{equation}
 \end{enumerate}
\end{prop}

Letting $a_1\leq a_2 \leq\ldots\leq a_m$ be the eigenvalues of $A$, one obtains that the
a.c. spectrum of $H_0$ is given by the union of the bands
$\bigcup_{j=1}^m (\overline{I^{K,L}}+a_j)$.
However, as in previous work using this method  \cite{KS,Sad} we have to restrict ourselves to the intersections of such bands and we need to assume that this intersection is not empty.
Therefore, define
\begin{equation}
 I_A^{K,L}:=\bigcap_{j=1}^m\,I^{K,L}+a_j\;.
\end{equation}

For the regular tree-strip, the technique with resonances does not have this weakness (cf. \cite{Sha}) and in fact gives existence of a.c. spectrum in a set that corresponds to the 
$\ell^1$ spectrum of the free operator (intersected with the real line\footnote{The $\ell^1$ spectrum of the adjacency operator is in general not a subset of the real line}).
However, \cite{Sha} needs a full random matrix potential $V(x)$ as e.g. from the GOE ensemble and therefore does not handle the Anderson model on tree-strips.
Extending this method to a non regular tree or tree-strip such as the Fibonacci tree would be interesting in order to confirm that
the $\ell^1$ spectrum of the adjacency operator determines the mobility edge for small $\lambda$.
Besides we will also need to assume that the random potential is
almost surely bounded:

\begin{assumptions} 
The following assumptions turn out to be crucial for the results.
\begin{enumerate}
\item[{\rm (V)}] The distribution $\nu$ of $V(x)$ is compactly supported in $\Sym(m)$.
\item[{\rm (A)}] 
Assume that the eigenvalues of $A$ are such that the set $I_A^{K,L}$ is not empty in which case it is a union of finitely many open intervals. \\
By Proposition~\ref{prop-spectrum}~(iii) for any fixed $L$ and $A$ there is a $K_0$ such that for $K>K_0$, $I^{K,L}_A$ is not empty, so the condition is fulfilled.

In the Fibonacci case $K=L=1$ this assumption reduces to
$$
a_{\rm max} - a_{\rm min} < 3\sqrt{3}\;,
$$
where $a_{\rm max}$ is the biggest and $a_{\rm min}$ the smallest eigenvalue of $A$, and then
\begin{equation}
 I^{1,1}_A=\left(-\frac32\sqrt{3}+a_{\rm max}\,,\,\frac32\sqrt{3}+a_{\rm min} \right)
 \;\setminus\;\{a_1,a_2,\ldots,a_m \}\;.
\end{equation}
\end{enumerate}
\end{assumptions}

In order to consider the spectrum of $H_\lambda$ we introduce the matrix-valued spectral measures at the vertices of the forest $\TT^{K,L}$. 
For $x\in\TT^{K,L},\,j\in\I=\{1,\ldots,m\}$
let $|x,j\rangle$ denote the
element in $\ell^2(\TT,\CC^m)$ satisfying $|x,j\rangle (y) = \delta_{x,y} e_j$ where $e_j$ is the $j$-th canonical basis vector in $\CC^m$.
Similar as before, $\langle x,j| H |y,k\rangle$ denotes the scalar product between
$|x,j\rangle$ and $H|y,k\rangle$ with the convention that the scalar product is linear in the second and anti-linear in the first component.
Then, for $x\in \TT^{K,L}$ we define the random, positive matrix valued measure $\mu_x$ on $\RR$ by
\begin{equation}
 \int f(E)\,d\mu_x(E) = \left[ \,\langle x,j |f(H_\lambda)|x,k\rangle\, \right]_{j,k\in\I}\;
\end{equation}
for all compactly supported, continuous functions $f$ on $\RR$.

\begin{theorem} \label{th:main} 
Let the assumptions {\rm (A)} and {\rm (V)} be satisfied.
Moreover, for $p\neq 0$ let $1^{(p)}_0$ be the first vertex (smallest distance to root) with 
label $0$, i.e. $1^{(p)}_0$ is the unique vertex with $d(0^{(p)},1^{(p)}_0)=L-p+1$.
Then, there is an open neighborhood $U$ of $\{0\}\times I^{K,L}_{A}$ 
in $\RR^2$ such that the following holds:
\begin{enumerate}[{\rm (i)}]
\item The spectrum of $H_\lambda$ is almost surely
purely absolutely continuous in $U_\lambda=\{E:(\lambda,E)\in U\}$.
\item For every $x\in\TT^{K,L}_0$ and any $x\in\TT^{K,L}_p$ with $p=1,\ldots,L$ and
$d(x,0^{(p)})> L-p$, 
the density of the absolutely continuous average spectral measure $\E(\mu_{x})$ in $U_\lambda$
depends continuously on $(\lambda,E)\in U$.
The condition $d(x,0^{(p)})>L-p$ means that either $x=1^{(p)}_0$ or that $x$ is a
descendant of $1^{(p)}_0$ in $\TT^{K,L}_p$, $p\neq 0$.
\item The density of $\E(\mu_{0^{(0)}})$ and $\E(\mu_{1^{(p)}_0})$ for $p=1,\ldots,L$ 
are positive definite in $U_\lambda$.
This implies that $H_\lb$ and also all the parts $H_\lb^{(p)}$ on $\ell^2(\TT^{K,L}_p \times \I)$ have spectrum in $U_\lambda$ with positive probability.
\end{enumerate}
\end{theorem}

\begin{remark} $ $
\label{rem-noremove}
Except for the removal of $\Ee_L$, the set $I^{K,L}_A$ here corresponds to $I_{A,S}=I_{A,S^{K,L}}$ in \cite{Sad}. The handling of the line segments requires to exclude the finitely many energies in $\Ee_L$.
In \cite{Sad} we also needed to remove some more energies defined by a very technical condition to get some smaller set $\hat I_{A,S}$.
This set occurred because the main argument is based on the Implicit Function Theorem and one needs a certain Frechet derivative
to be invertible. Here, this Frechet derivative will always be invertible for all 
$E\in I^{K,L}_A$ by Lemma~\ref{lem-not-1}
which relies on the identity given in Proposition~\ref{prop-Gamma}.
This identity holds more general for the Green's functions on trees of finite cone type as shown in Appendix~\ref{app-gamma}. Using this identity and essentially the same line of arguments as in Lemma~\ref{lem-not-1}, one can actually show that $I_{A,S}=\hat I_{A,S}$ as defined in \cite{Sad} for any 
substitution matrix $S\in \ZZ_+^{k\times k}$ that has at least one positive entry on the diagonal and where the determinants of all diagonal minors of $-S$ are non-positive.
Here a diagonal minor of a square matrix is a matrix obtained 
by deleting finitely many rows and the {\bf same} columns, i.e. one deletes the
$i_1,\ldots, i_l$ row and column to get a $(k-l)\times (k-l)$ matrix. 
The $1\times 1$ diagonal minors are exactly the diagonal elements.
For $2\times 2$ substitution matrices this condition simply means that there is one positive diagonal element and the determinant is negative.
\end{remark}

The important objects we work with are the matrix Green's functions given by
\begin{equation}\label{eq-def-G^x}
 G_\lambda^{[x]}(z):=\left[\, \langle x,j|(H-z)^{-1} | x,k\rangle\,\right]_{j,k\in\I}\,\in\,{\CC^{m\times m}}
\end{equation}
for $\im(z)>0$.
The most important ingredient to obtain Theorem~\ref{th:main} is the following.

\begin{theorem}\label{th:main2}
Under assumptions {\rm (V)} and {\rm (A)} 
there exists an open neighborhood $U$ of $\{0\}\times I^{K,L}_{A}$ in $\RR^2$ 
such that for all vertices $x\in\TT^{K,L}_0$ and all $x\in\TT^{K,L}_p$ with 
$d(0^{(p)},x)>L-p$ the functions
\begin{align} 
(\lambda,E,\eta) &\mapsto \E\left(G_\lambda^{[x]}(E+i\eta)\right)\;, \notag \\
(\lambda,E,\eta) &\mapsto \E\left(\left|G_\lambda^{[x]}(E+i\eta)\right|^2\right)\;, \notag
\end{align}
defined for $\eta>0$, have continuous extensions to $U\times [0,\infty)$.
\end{theorem}

Let us show now that Theorem~\ref{th:main2} implies Theorem~\ref{th:main}.
\renewcommand{\proofname}{Proof of Theorem~\ref{th:main}}
\begin{proof}
Part (ii) follows immediately as
$\E(G_\lambda^{[x]}(z))$ is the Stieltjes transform of $\E(\mu_{x})$ and hence
$\E (\mu_{x})(dE)=*-\lim_{\eta\downarrow 0} \frac1\pi \im \E(G_{\lambda,E+i\eta})\,dE$.
Here $*-\lim$ denotes a limit in the weak $*$ topology on bounded measures.
For part (iii) we remark that for $\lb=0$ one can explicitly calculate the matrix Green's functions and see that for energies $E\in I^{K,L}_A$ the limit of the 
imaginary parts are positive definite (cf. Remark~\ref{rem-I_A}). By continuity around $\lb=0$ this remains true for a possibly smaller neighborhood $U\supset \{0\}\times I^{K,L}_A$.

To get (i) note that for any compact interval $[a,b]\subset U_\lb$ one has
by Fatou's lemma
$$\E \Big(\liminf_{\im(z) \downarrow 0}\int_{a}^b \Tr\big(\big|G^{[x]}_{\lb}\, (z)\big|^2\big)\,dE\Big)
\leq \liminf_{\im(z) \downarrow 0} \int_{a}^b \E \Big(\Tr\big(\big|G^{[x]}_{\lb}\, (z)\big|^2\big)\Big) \,dE
\,<\,\infty
$$
Thus, 
\beq\label{limprobone}
\liminf\limits_{\im(z) \downarrow 0}\int_{a}^b \Tr\left(\left|G^{[x]}_{\lb}\, (z)\right|^2\right)\,dE \,<\,\infty \quad
\text{with probability one}. 
\eeq
As $G_\lb^{[x]}(z)$ is the Stieltjes transform of $\mu_{x}$ this implies
that, almost surely, $\mu_{x}$ is absolutely continuous with respect to the Lebesgue measure in 
$(a,b)\subset U_\lambda$ and the density is a positive matrix valued $L^2$ function for all $x\in\TT^{K,L}_0$
and all $x \in \TT^{K,L}_p$ with $d(0^{(p)},x)>L-p$.
(cf. \cite[Theorem 4.1]{Kl6} and \cite[Theorem 2.6]{Kel}). 
This gives the almost sure pure a.c. spectrum on any interval $(a,b)$ with closure
in $U_\lb$ as the cyclic spaces of the vectors $|x,j\rangle$ for those $x$ span $\ell^3(\TT^{K,L}\times\I)$. 
Writing $U_\lb$ as a countable union of such closed compact intervals $[a,b]$ one realizes by taking the intersection of the corresponding sets in the probability space that
with probability one the spectrum of $H_\lb$ is purely absolutely continuous in $U_\lb$.
\end{proof}

\renewcommand{\proofname}{Proof}

\vspace{.2cm}

The paper is structured as follows. In Section~\ref{sec-unper} the Green's functions $\Gamma^{(p)}_z$ are considered in
more detail and Proposition~\ref{prop-spectrum} is proved.
In Section~\ref{sec-rec} we consider the recursion of the forward Green's matrices that lead to a fixed point equation.
Then, in Section~\ref{sec-ban}, we introduce the important Banach spaces in which this fixed point equation has to be analyzed and
in Section~\ref{sec-frechet} we investigate the Frechet derivative.
Finally, in Section~\ref{sec-proofs} we conclude to obtain Theorem~\ref{th:main2}.
Appendices \ref{app-gamma}, \ref{app-FMF} and \ref{app-compact} state some general facts that are used along the way.


\section{The unperturbed Green's functions \label{sec-unper}}

The main goal of this section is to prove Proposition~\ref{prop-spectrum}. 
Some of the statements are consequences of more general considerations as 
done in \cite{Kel, KLW}. 

Recall that we denoted the Green's functions at the roots $0_p$ by $\Gamma^{(p)}_z$
for $\im(z)>0$ and for $z=E\in I_p^{K,L}$ it also denotes the limit for $\im(z)\downarrow 0$
if it exists (cf. \eqref{eq-def-Gamma}, \eqref{eq-def-Gamma-E}).
Let us also define the spectral measures $\nu^{(p)}$  by
\begin{equation}
 \int f(x)\,d\nu^{(p)}(x)\,=\, \langle 0^{(p)} | f(\Delta) | 0^{(p)} \rangle\;.
\end{equation}
As $z\mapsto \Gamma_z^{(p)}$ is the Stieltjes transform of the measure $\nu^{(p)}$ the 
absolutely continuous part is the closure of $I^{K,L}_p$.
Now for any vertex $x\in\TT^{K,L}_p$ we can take the path to the root and cut off the trees
connected to this path, which are all equivalent to one of the $\TT^{K,L}_q$.
Then, using an induction argument (induction over distance to root) as done in 
\cite[Prop. 2.9]{Kel} or alternatively using \cite[Lemma 2.2]{FLSSS} one finds
\begin{equation}
 \sigma_{ac}(\Delta)\,=\,\bigcup_{p=0}^L \overline{I^{K,L}_p}\,.
\end{equation}

Note that when cutting the connections to the root in $\TT^{K,L}_p$ then for $0<p<L$ we get $\TT^{K,L}_{p+1}$, for $p=L$ we get $\TT^{K,L}_0$ and for $p=0$ we get the union of $K$ times $\TT^{K,L}_0$ and one $\TT^{K,L}_1$.
Therefore, the standard recursion relations for the Green's functions that can be obtained from the resolvent identity
(cf. \cite{ASW,FHS,Kl3,KLW,Kel}) are given by
\begin{equation}\label{eq-rec-Greens}
 \Gamma^{(0)}_z = \frac{-1}{z+K\Gamma^{(0)}_z+\Gamma^{(1)}_z}\;,\quad
 \Gamma^{(p)}_z = \frac{-1}{z+\Gamma^{(p+1)}_z}\;,\quad
 \Gamma^{(L)}_z = \frac{-1}{z+\Gamma^{(0)}_z}\;,
\end{equation}
where $p=1,\ldots,L-1$.
By an hyperbolic contraction argument as in \cite{KLW,Kel} one obtains that for $\im(z)>0$, 
these equations combined with the restriction $\im(\Gamma^{(p)}_z)>0,\,p=0,1,\ldots,L$
determine the Green's functions uniquely.

The right hand sides of the last two equations in \eqref{eq-rec-Greens} can be seen as M\"obius actions $\smat{0 & -1 \\ 1 & z}\cdot \Gamma=\frac{-1}{\Gamma+z}$.
Defining the polynomials $c_0(z)=0,\,c_1(z)=1$ and iteratively 
$c_{k+1}(z)=z c_k(z)-c_{k-1}(z)$ one obtains by induction
\begin{align}\label{eq-T-powers}
 & \pmat{0 & -1 \\ 1 & z}^{k}\,=\,\pmat{-c_{k-1}(z) & -c_k(z) \\ c_k(z) & c_{k+1}(z)}\;,\quad\text{and} \\
 & c_k(z)=\frac{c_+(z)^k - c_-(z)^k}{c_+(z)-c_-(z)}\qtx{where} c_{\pm}(z)=\frac12\left(z\,\pm\,
 \sqrt{z^2 -4}\right)\;.
\end{align}
Here one can choose any of the two roots for $\sqrt{z^2-4}$, changing the root switches $c_+(z)$ and $c_-(z)$ but does not change
$c_k(z)$.
As the determinant of these matrices are always one, we also find
\begin{equation}\label{eq-c_k-det}
c_{k}^2(z)-c_{k-1}(z) c_{k+1}(z)\,=\,1 \;.
\end{equation}
It can be seen from the recursion relations defining $c_k(z)$ that $c_k(z)$ is a polynomial of degree $k-1$ with real coefficients and it is even if $k$ is odd and odd if $k$ is even.

Now, \eqref{eq-rec-Greens} and \eqref{eq-T-powers} lead to
\begin{equation}
 \Gamma^{(0)}_z \left[z+K\Gamma^{(0)}_z+ \pmat{-c_{L-1}(z) & -c_L(z) \\ c_L(z) & c_{L+1}(z)} \cdot \Gamma^{(0)}_z \right]\;+1\;=\;0
\end{equation}
which using $zc_L-c_{L-1}=c_{L+1}$ can be rewritten as
\begin{equation}
 K c_L(z) (\Gamma^{(0)}_z)^3 + (K+1) c_{L+1}(z)\, (\Gamma^{(0)}_z)^2+ z c_{L+1}(z)\, \Gamma^{(0)}_z +c_{L+1}(z)
 \,=\,0
\end{equation}
Using the uniqueness of solutions in the upper half plane we see that $I_0^{K,L}$ is given by the set of real energies $E$ such that the cubic equation with real coefficients
\begin{equation}\label{eq-cubic-eq}
K c_L(E)\,x^3 + (K+1) c_{L+1}(E)\,x^2+E c_{L+1}(E)\,x+ c_{L+1}(E)\,=\,0
\end{equation}
has a solution $x$ with a positive imaginary part and $\Gamma^{(0)}_E$ is that solution.

With this characterization we have now established everything to prove Proposition~\ref{prop-spectrum}.

\renewcommand{\proofname}{Proof of Proposition~\ref{prop-spectrum}}

\begin{proof}
Using equations \eqref{eq-rec-Greens} in the limits $\im(z)\to 0$ we see that 
\begin{equation}\label{eq-I_p-equal}
I^{K,L}_0 \subset I^{K,L}_1 \subset \cdots \subset I^{K,L}_L \subset I^{K,L}_0\;, \qtx{hence}
I^{K,L}_p=I^{K,L}_0.
\end{equation}
For the inclusion $I^{K,L}_0 \subset I^{K,L}_1$ note that the first equation of \eqref{eq-rec-Greens} gives $\Gamma^{(0)}_z=\frac{1}{2K} \left(-\Gamma^{(1)}_z - z + \sqrt{(\Gamma^{(1)}_z+z)^2 - 4K} \right)$ where we have to take the square root with the positive imaginary part (the other one will have a negative imaginary part). Hence for $E\in I^{K,L}_1$ we have a limit
for $\Gamma^{(0)}_{E+i\eta}$ for $\eta\downarrow 0$ with non-negative imaginary part. As 
$\im(\Gamma^{(1)}_E)>0$ we find $\im(\Gamma^{(0)}_E)>0$ for $E\in I^{K,L}_1$.

By the characterization of $I^{K,L}_0$ as in \eqref{eq-cubic-eq} one sees that $I^{K,L}_0$ is in fact the set of energies $E$ where the discriminant $D=D(E)$ of \eqref{eq-cubic-eq}
is negative,
\begin{align}
 D & =c_{L+1}^2\Big( 18K(K+1)Ec_L c_{L+1} +(K+1)^2 E^2 c_{L+1}^2-4K E^3 c_L c_{L+1} \notag \\
 & \qquad\qquad-27 K^2 c_L^2\,-4(K+1)^3 c_{L+1}^2 \Big)\;.\label{eq-discr}
\end{align}
If $c_{L+1}(E)=0$ then the term inside the parenthesis for $D$ is negative coming from the term $-27K^2 c_L^2(E)$, using \eqref{eq-c_k-det}, $c_L(E)\neq 0$. By induction one also gets that $c_{L+1}(E)$ is the characteristic polynomial 
$\det(E\one-\Delta_L)$ for the matrix $\Delta_L$ as in \eqref{eq-def-Delta-d}, so it has real roots and in a neighborhood
of these roots, $D<0$. Therefore, $I^{K,L}_0$ is not empty.
Moreover, as $D(E)$ is a polynomial in $E$, $I^{K,L}_0$ and $I^{K,L} \setminus \Ee_L$ are unions of finitely many open intervals.
As there is a solution formula for cubic equations we also find that $\Gamma^{(0)}_E$ and by \eqref{eq-rec-Greens}
all $\Gamma^{(p)}_E$ are analytic in $E \in I^{K,L}_0$. This finishes part (i).

As $I^{K,L}$ has only finitely many points removed from $I^{K,L}_0$, we have
\begin{equation}
 \overline{I^{K,L}}=\overline{I^{K,L}_0}=\bigcup_{p=0}^L \overline{I^{K,L}_p}\,=\,\sigma_{ac}(\Delta)\;
\end{equation}
giving part (ii).

To get part (iii) note that we find an open neighborhood
$\Oo$ of these zeros, such that for fixed $L$, all $E\in \Oo$ and all $K>0$, 
\begin{equation}
18 K(K+1) E c_L c_{L+1}+(K+1)^2 E^2 c_{L+1}^2 - 4K E^3 c_L c_{L+1} - 27 K^2 c_L^2\,<\,0\;.
\end{equation}
In particular, if $E\in \Oo$ and $c_{L+1}(E)\neq 0$ then for all $K$ we find $D(E)<0$.
In the compact set $[-E_0, E_0] \setminus \Oo$, $c^2_{L+1}(E)$ attains a minimum value bigger
than zero and $c_L$ and $c_{L+1}$ are bounded. The highest power in $K$ appearing inside the parenthesis on the right hand side of \eqref{eq-discr} is the negative term $-4(K+1)^3 c_{L+1}^2$. 
Therefore, we find $K_0=K_0(L,E_0)$, such that $D(E)<0$ for all
$E\in [-E_0,E_0]\setminus \Oo $ and $K>K_0$.
In this case we obtain $[-E_0,E_0] \subset \overline{I^{K,L}_0}$ finishing part (iii).

For part (iv) note that for $K=L=1$ using $c_1(E)=1,\,c_2(E)=E$ we find
$ D= 4E^4-27E^2=E^2(4E^2-27)$ and $\Ee_1=\{0\}$ giving
\begin{equation}
I^{1,1}=I^{1,1}_0=\left(-\frac{3\sqrt{3}}{2}\,;\,\frac{3\sqrt{3}}{2}\right) \,\setminus\,\{0\}\;. 
\end{equation}
\end{proof}

\renewcommand{\proofname}{Proof}

For our further analysis we will also need the following property.

\begin{prop}\label{prop-Gamma}

For $E\in I^{K,L}$ one has 
\begin{equation}\label{eq-gamma-rels}
K\,\left|\Gamma^{(0)}_E\right|^2\,+\, \left|\Gamma^{(0)}_E \Gamma^{(1)}_E \cdots \Gamma^{(L)}_E\right|^2\,=\,1\;.
\end{equation}
\end{prop}

\begin{proof}
Equation \eqref{eq-gamma-rels} can be rewritten as
\begin{equation}
\det\left[\one\,-\,
 \begin{pmatrix}
 \left|\Gamma^{(0)}_E\right|^2 &  \\ & \ddots & \\ & & \left|\Gamma^{(L)}_E\right|^2
 \end{pmatrix}
 S^{K,L}\, \right]\,=\,0
\end{equation}
which is an identity shown for the more general case of trees of finite cone type in 
Proposition~\ref{prop-gr} in Appendix~\ref{app-gamma}.
\end{proof}

\begin{remark}\label{rem-I_A}
As $I^{K,L}_A=\bigcap_{j=1}^m (I^{K,L} +a_j)$ we get for $E\in I^{K,L}_A$ that $E-a_j \in I^{K,L}$ for all eigenvalues $a_j$ of $A$. In particular $\Gamma^{(p)}_{E-a_j}$ exists and $\im(\Gamma^{(p)}_{E-a_j})>0$.
\end{remark}


\section{Recursion for matrix Green's functions} \label{sec-rec}

A key identity for the analysis of Schr\"odinger operators on trees and tree-strips are the well known identities like
\eqref{eq-rec-Greens} for forward matrix Green's functions
that can be obtained from the resolvent identity.
Let $\TT_x$ be some tree that has a vertex $x$ and let $\Nn(x)$ be the set of neighboring points in $\TT_x$. 
If $H_\lb$ is given as in \eqref{eq-H_lb} on $\ell^2 (\TT_x\times\I)$ with $\I=\{1,\ldots,m\}$ then
the Green's matrix $(G_{x,\lb,z})_{jk} = \langle x,j|(H_\lb-z)^{-1}|x,k\rangle$ satisfies
\begin{equation}\label{eq-G-rec1}
 G_{x,\lb,z}\,=\, - \left(\sum_{y\in\Nn(x)} G_{\lb,z}^{(y|x)} +z-A-\lambda V(x)\right)^{-1}\;.
\end{equation}
Here and in many equations below the upper index $(y|x)$ indicates that we look at the vertex $y$ and remove the branch of the tree that is emanating from $y$ and going through $x$.
This means we let $\TT_x^{(y|x)}$ denote the tree with vertex $y$ where the branch going from $y$ to $x$ is removed (i.e. the tree of vertices $x'$ satisfying  
$d(x,x')=d(x,y)+ d(y,x')$\, ).
Furthermore, $H_\lb^{(y|x)}$ is the operator $H_\lb$ restricted to $\ell^2(\TT_x^{(y|x)} \times \I)$ with Dirichlet boundary conditions and
$(G_{\lb,z}^{(y|x)})_{jk}=\langle y,j|(H_\lb^{(y|x)}-z)^{-1}|y,k\rangle$.
This equation is valid in any tree, in particular it is valid in subtrees when certain branches were cut.

Now if a tree $\widetilde \TT$ has an assigned root $0\in\widetilde \TT$ and we set $\TT_x:=\widetilde\TT^{(x|0)}$ to be the tree where we disconnect the branch at $x$ going through the root
then for $y\in \TT_x$ one finds $\TT_x^{(y|x)}=\TT_y$. The corresponding Green's matrices are $G_{x,\lb,z}=G^{(x|0)}_{\lb,z}$ and $G^{(y|x)}_{\lb,z}=G^{(y|0)}_{\lb,z}$ for $y\neq x, y\in\TT_x$ and they
only depend on the matrix potential on the branches at $x$ or $y$, respectively,
that go away from the root. Therefore we call them {\it forward} matrix Green's functions. 
Then,  \eqref{eq-G-rec1} becomes
\begin{equation}\label{eq-G-rec2}
 G^{(x|0)}_{\lb,z}\,=\,- \left(\sum_{y\in\Nn(x)} G^{(y|0)}_{\lb,z} +z-A-\lambda V(x)\right)^{-1}\
\end{equation}
where $\Nn(x)$ is the set of neighbors of $x$ in the tree $\TT_x=\widetilde\TT^{(x|0)}$, i.e. the set of forward neighbors (or children).

In the case studied in this paper we consider the trees $(\TT^{K,L}_p)_x=(\TT^{K,L}_p)^{(x|0^{(p)})}$ for $x\in\TT^{K,L}_p$. Now if $l(x)$ is the label of $x$, then, $(\TT^{K,L}_p)_x \cong \TT^{K,L}_p$ and the distribution of $G^{(x|0^{(p)})}_{\lb,z}$ (for $x\in\TT^{K,L}_p$) only depends on $\lb,\,z$ and the label $l(x)$.
Moreover, the different $G^{(y|0^{(p)})}_{\lb,z}$ occurring on the right hand side are independent and they are independent of $V(x)$. 
Therefore one might want to work with some averaged quantities. However, the occurring inverse on the right hand side of \eqref{eq-G-rec2}
prevents one from getting something useful by just applying the expectation to this equation. The key idea is now to represent the operation $G\mapsto -G^{-1}$ as a linear
operator in some function space. More precisely, as in \cite{KS, KS2, Sad} we associate to
symmetric $m\times m$ matrices $G$ with positive imaginary part the following functions:
Let $\Sym^+(m)$ denote the real, positive semi-definite matrices, i.e. $M\in\Sym^+(m) \Leftrightarrow M\in\Sym(m),\,M\geq 0$, and for
$\re(G)\in\Sym(m),\; \im(G)\in \Sym^+(m)$ we define the bounded functions
\begin{align}\label{eq-def-zeta}
\zeta_G&\,:\,\Sym^+(m)\to\CC\;,\;\ze_G(M) = e^{\frac{i}2 \Tr(GM)} \\ \label{eq-def-xia}
\xi_G&\,: \,(\Sym^+(m))^2\to\CC\;,\;\xi_G(M_+,M_-) = \zeta_G(M_+) \overline{\zeta}_G(M_-)=e^{\frac{i}2 \Tr(GM_+)-\frac{i}2\Tr(\overline{G}M_-)}
\end{align}
Note that if $\im(G)>0$ then $\zeta_G$, $\xi_G$ and all its derivatives are exponentially decaying.
A supersymmetric Fourier transform $T$ and $\Tt=T\otimes T$ as defined in \cite{KS, KS2, Sad} gives for $\im(G)>0$ that
$T\zeta_G = \zeta_{-G^{-1}}$ and $\Tt \xi_G=\xi_{-G^{-1}}$. 
For the convenience of the reader these operators and all important Banach spaces as used in previous works will be defined precisely
in the next section.

Using these linear operators, \eqref{eq-G-rec2} can be rewritten as
\begin{equation}\label{eq-G-rec3}
 \zeta_{G^{(x|0)}_{\lb,z}}=T \Big(\zeta_{z-A-\lb V(x)} \prod_{y\in\Nn(x)} \zeta_{G^{(y|0)}_{\lb,z}}\Big)\;,\:
 \xi_{G^{(x|0)}_{\lb,z}}=\Tt \Big(\xi_{z-A-\lb V(x)} \prod_{y\in\Nn(x)} \xi_{G^{(y|0)}_{\lb,z}}\Big)\;.
\end{equation}
For $p\in\{0,1,\ldots,L\}$ we take some $x\in\TT^{K,L}_q$ with label $l(x)=p$ let $0=0^{(q)}$ and define
\begin{equation}
 \zeta^{(p)}_{\lb,z} := \E(\zeta_{G^{(x|0)})_{\lb,z}})\;,\qquad
  \xi^{(p)}_{\lb,z} := \E(\zeta_{G^{(x|0)}_{\lb,z}})\;.
\end{equation}
Then, taking expectations in \eqref{eq-G-rec3} gives for $\im(z)>0$,
$1\leq p\leq L-1$
\begin{align} \label{eq-zeta-recursion}
  & \zeta^{(0)}_{\lb,z}=TB_{\lb,z} \left((\zeta^{(0)}_{\lb,z})^K \zeta^{(1)}_{\lb,z}\right) \;,\quad \zeta^{(p)}_{\lb,z}=TB_{\lb,z} \zeta^{(p+1)}_{\lb,z}\;,\quad \zeta_{\lb,z}^{(L)}=TB_{\lb,z} \zeta_{\lb,z}^{(0)} \;\\  \label{eq-xi-recursion}
  & \xi^{(0)}_{\lb,z}=\Tt\Bb_{\lb,z} \left((\xi^{(0)}_{\lb,z})^K \xi^{(1)}_{\lb,z}\right) \;,\quad \xi^{(p)}_{\lb,z}=\Tt\Bb_{\lb,z} \xi^{(p+1)}_{\lb,z}\;,\quad \xi_{\lb,z}^{(L)}=\Tt\Bb_{\lb,z} \xi_{\lb,z}^{(0)} \;
\end{align}
where $B_{\lb,z}$ or $\Bb_{\lb,z}$ are the multiplication operators defined by
\begin{align} \label{eq-def-B}
 B_{\lb,z} f(M) &= \E (\zeta_{z-A-\lb V(x)}(M)) f(M)\\
 \Bb_{\lb,z} g(M_+,M_-) &= \E(\xi_{z-A-\lb V(x)}(M_+,M_-))\; g(M_+,M_-)\;.
\end{align}
Recall that we assume without loss of generality that $A=\diag(a_1,\ldots,a_m)$ is diagonal.
Therefore, in the free case $\lambda=0$ one obtains
\begin{equation}
\ze^{(p)}_{0,z} = \ze_{A^{(p)}_z}\;,
\;\; \xi^{(p)}_{0,z} = \xi_{A^{(p)}_z}\;,\;\;\text{where}
\;\;
A^{(p)}_z := \diag(\Gamma^{(p)}_{z-a_1},\Gamma^{(p)}_{z-a_2},\ldots,\Gamma^{(p)}_{z-a_m})\;.
\label{eq-ze0}
\end{equation}
For  $E\in I^{K,L}_{A}$ the point-wise limits 
\begin{align}
& \ze^{(p)}_{0,E} {:=} \lim_{\eta \downarrow 0} \ze^{(p)}_{0,E+i\eta} 
=\ze_{A^{(p)}_E}\, , \label{eq-ze00} \qquad
\xi^{(p)}_{0,E} {:=} \lim_{\eta \downarrow 0} \xi^{(p)}_{0,E+i\eta} =
\xi_{A^{(p)}_E}
\end{align}
exist for $p\in\{0,\ldots,L\}$, where
\begin{equation}
 A_E^{(p)}=\lim_{\eta\downarrow 0} A_{E+i\eta}^{(p)}\;,\qquad
[A_E^{(p)}]_{jk}=\delta_{jk}\,\Gamma^{(p)}_{E-a_j} \label{eq-def-AE}
\end{equation}
are diagonal $m\times m$ matrices with strictly positive imaginary part.
The important point is to understand equations \eqref{eq-zeta-recursion} and \eqref{eq-xi-recursion} as fixed point equations in appropriate Banach spaces.


\section{The proper Banach spaces} \label{sec-ban} 
  
Let us first briefly introduce the important Banach spaces as in \cite{KS, KS2, Sad} and for the readers convenience we will
list all important notation and definitions from previous works in the next definition. In particular we will also give the precise definitions of the operators $T$ and $\Tt$
mentioned above.
All proofs and arguments are omitted as \cite[Section 3]{Sad} uses the exact same notations and has the statements in more detail.
We will also skip to mention the connection to supersymmetry and how these spaces naturally evolve from this formalism.
For a supersymmetric background see \cite[Appendix~B]{Sad} or \cite{KS}.
As above we set $\I=\{1,\ldots,m\}$.

\begin{defini} \label{def-spaces}
\noindent
\begin{enumerate}[{\rm (a)}]
\item $\PP(\I)=\{a : a\subset \I\}$ denotes the set of all subsets of $\I$
\item $\Pp$ denotes the set of pairs $(\bar a, a)$ of subsets of $\I$ with the same cardinality,
$ \Pp:=\{(\bar a, a)\,: \bar a, a \subset \I,\,|\bar a|=|a| \}\;$
\item $n=n(m)$ will be the smallest integer such that $n\geq\frac{m}2$
\item With $\bar\aaa=(\bar a_1,\ldots,\bar a_n),\, \aaa=(a_1,\dots, a_n)\,\in(\PP(\I))^n$ define
\begin{equation}
\Pp^n:=\{(\bar\aaa,\aaa) \in (\PP(\I))^n \times (\PP(\I))^{n}\,:\,(\bar a_l, a_l)\in \Pp \}\;.
\end{equation}
We also let $|\aaa|:=|a_1|+|a_2|+\ldots+|a_n|$ and $\aaa^\brc:=(a_1^c,\ldots, a_n^c)$ with $a_j^c=\I\setminus a_j$.
\item For functions on $(\Sym^+(m))$ let $\partial_{j,k}$ denote the derivative with respect to the $j,k$- entry of $M$,
$\partial_{j,k} f(M)= \frac{\partial}{\partial M_{j,k}} f(M)$ 
(by symmetry, $\partial_{j,k}=\partial_{k,j}$) and let $\tilde \partial_{j,k}= \frac12\partial_{j,k}$ for $j\neq k$ and $\tilde \partial_{j,j}=\partial_{j,j}$.
\item For $(\bar a, a) \in \Pp$ with $a=\{k_1,\ldots,k_c\},\;k_1 < k_2 \ldots < k_c;\;
\bar a =\{\bar k_1,\ldots,\bar k_c\},\;\bar k_1 < \bar k_2 < \ldots < \bar k_c$,
define
\begin{equation}
 \bpart_{\bar a, a}:=\begin{pmatrix}
       \tilde \partial_{\bar k_1,k_1} & \cdots & \tilde \partial_{\bar k_1,k_c}  \\
        \vdots & \ddots & \vdots \\
	\tilde \partial_{\bar k_c,k_1} & \cdots & \tilde \partial_{\bar k_c, k_c} \end{pmatrix}\;,\quad
	D_{\bar a,a}:=\det(\bpart_{\bar a, a})\;
\end{equation}
with the convention that $D_{\emptyset,\emptyset}$ is the identity  operator.
\item For $(\bar\aaa, \aaa)\in\Pp^n$, $\bar\aaa=(\bar a_1,\ldots,\bar a_1),\, \aaa=(a_1,\dots, a_n)$, let
$D_{\bar\aaa,\aaa}:=\prod_{\ell=1}^n D_{\bar a_\ell, a_\ell}\;$.
There is a function $\sgn(\bar\aaa,\aaa,\bar\bbb,\bbb,\bar\bbb',\bbb')
\in \{-1,0,1\}$ such that 
\begin{equation} \label{eq-Leibn0}
D_{\bar\aaa,\aaa}\,(fg) = 
\sum_{ (\bar\bbb,\bbb),(\bar\bbb',\bbb')\in\Pp^n}\;\sgn(\bar\aaa,\aaa,\bar\bbb,\bbb,\bar\bbb',\bbb')
(D_{\bar\bbb,\bbb}\, g)\,(D_{\bar\bbb',\bbb'}\, f)\;.
\end{equation}
\item For $f\in C^\infty(\Sym^+(m)),\;g\in C^\infty((\Sym^+(m))^2)$ and $p\geq 1$ we introduce the norms $\hn f\hn_p$ and $\hnn g \hnn_p$ by
\begin{equation} \label{eq-def-norms}
\hn f \hn_p^2\;: =\;
\sum_{(\bar\aaa,\aaa)\in\Pp^n}\; 4^{|\aaa|}\left[\int_{\RR^{m\times 2n}} \left|
 D_{\bar\aaa,\aaa} \;f\,(\bvp\bvp^\top)\right|^p d^{2mn}\bvp\right]^{2/p}
\end{equation}
where $\bvp$ denotes a $m\times 2n$ matrix, $\bvp\in\RR^{m\times 2n}$, and
\begin{align}
\hnn g\hnn_p^2 := &
\sum_{(\bar\aaa,\aaa),(\bar\bbb,\bbb) \in \Pp^n} 4^{|\aaa|+|\bbb|} \\
& \left[ \int_{(\RR^{m\times 2n})^2}
\left|  D_{\bar\aaa,\aaa}^{(+)} D_{\bar\bbb,\bbb}^{(-)} g(\bvp_+\bvp_+^\top,\bvp_-\bvp_-^\top)
\right|^p d^{2mn}\bvp_+ d^{2mn}\bvp_-\;\right]^{2/p}.
\end{align}
Here, $D_{\bar\aaa,\aaa}^{(\pm)}$ denotes the operator $D_{\bar\aaa,\aaa}$ with respect to the entry $M_\pm=\bvp_\pm\bvp_\pm^\top$.
Note that the map $\RR^{m\times 2n} \ni \bvp \mapsto \bvp\bvp^\top\in\Sym^+(m)$ is surjective as $n\geq 2m$.
We also define the corresponding norms $\hn f\hn_\infty$ and $\hnn g \hnn_\infty$ using the limit $p\to\infty$ which are given by the sums of the corresponding suprema.
\item Let $\Sym_\CC(m)$ denote the complex symmetric, $m\times m$ matrices and for $B\in\Sym_\CC(m)$ with strictly positive imaginary part
(i.e., $\im B > 0$), let $\PE(B)$ denote the vector space spanned by functions of the form
$f(M) = p(M) \zeta_B(M)$, where $p(M)$ is a polynomial in the entries of $M\in\Sym^+(m)$. Clearly, due to the exponential decay in
$\zeta_B$ one has $\hn f\hn_p<\infty$ for all $p\in[1,\infty]$.
\item Define $\PE(m)$ as the smallest vector space containing all vector spaces $\PE(B)$ for all $B\in\Sym_\CC(m)$ with $\im(B)>0$.
\item For $1\leq p\leq \infty$, let $\Ll^p$ be the completion of $\PE(m)$ with respect to the norm $\hn f \hn_p$. We furthermore set $\Hh=\Ll^2$, define $\hn f \hn_{2,p}:=\hn f\hn_2+\hn f\hn_p$ and let $\Hh_p$ be the completion of $\PE(m)$ w.r.t. to $\hn\cdot\hn_{2,p}$.
\item For $1\leq p\leq \infty$ let $\widehat \Ll^p$ be the completion 
of $\PE(m)\otimes\PE(m)$ w.r.t. $\hnn g \hnn_p$, set $\Kk=\widehat\Ll^2$ and
$\hnn g\hnn_{2,p}:=\hnn g \hnn_2+\hnn g \hnn_p$, 
and let $\Kk_p$ be the completion of $\PE(m)\otimes \PE(m)$ w.r.t. $\hnn g \hnn_{2,p}$.
In fact, $\Kk=\Hh\otimes \Hh$ as Hilbert space tensor product.
\item Let $\III\in\PP(\I)^n$  be given by $\III=(\I,\I,\ldots,\I)$.
\item The supersymmetric Fourier transform $T$ is given by
\begin{equation} \label{eq-def-T0}
Tf (\bvp'\bvp'^\top) := \frac{(-1)^{mn}}{\pi^{mn}} \, \int
e^{\pm i \Tr(\bvp' \bvp^\top)}\,D_{\III,\III} f(\bvp\bvp^\top)\, d^{2mn}\bvp\;.
\end{equation}
Note that the sign of $\pm$ does not matter by the symmetry $\bvp \to -\bvp$.
$\Tt=T\otimes T$ is given by
\begin{align}
& \Tt g(\bvp_+'{\bvp'}_+^\top, \bvp_-'{\bvp'}_-^\top) \\ \notag
& \quad =\frac1{\pi^{2mn}} \int e^{\pm i \Tr(\bvp_+'\bvp^\top_+ \pm \bvp_-' \bvp_-^\top)} 
D_{\III,\III}^+ D_{\III,\III}^- \,g(\bvp_+\bvp_+^\top,\bvp_-\bvp_-^\top)\,d^{2mn} \bvp_+\,d^{2mn} \bvp_-
\end{align}
where $D_{\III,\III}^\pm$ denotes the operator $D_{\III,\III}$ with respect to the entry $M_\pm=\bvp_\pm\bvp_\pm^\top$.
\end{enumerate}
\end{defini}
Roughly speaking, $\Ll^p$, $\widehat \Ll^p$ are the supersymmetric analogs
of $L^p$ spaces, $\Hh_p$ and $\Kk_p$ resemble $L^2 \cap L^p$ and the operators
$T$ and $\Tt$ have the role of the Fourier transform.
One might think that $\Hh_p=\Ll^2\cap\Ll^p$ as a set which is equivalent to saying that $\PE(m)$ is dense w.r.t. the $\hn\cdot\hn_{2,p}$ norm in $\Ll^2 \cap \Ll^p$.
Unfortunately, we were not able to prove this.
The spaces $\Ll^p,\, \widehat \Ll^p$ have not been used in former papers, but some
technical difficulties in this paper requires the use of $\Ll^\infty$ and $\Ll^1$ in Section~\ref{sec-proofs} for proving Theorem~\ref{th:main2}.

One finds (cf. \cite[eq. (B.22)]{Sad}, \cite[eq. (2.37)]{KS})
 \begin{equation}\label{eq-T-parts}
D_{\bar\aaa,\aaa} (Tf) = 
\tfrac{2^{mn}}{4^{|\aaa|}}\;\sgn(\aaa,\bar\aaa)\;\Ff(D_{\aaa^\brc,\bar\aaa^\brc}\,f) \qtx{for all} (\bar\aaa,\aaa) \in \Pp^n\;,
\end{equation}
where $\Ff f(\bvp\bvp^\top)$ denotes the Fourier transform w.r.t. to $\bvp$ in $\RR^{m\times 2n} \cong \RR^{2mn}$
and $\sgn(\aaa,\bar\aaa)\in\{-1,1\}$ is some sign.

Using the fact that the Fourier transform $\Ff f(\bvp\bvp^\top)$ maps $\PE(m)$ to $\PE(m)$ and continuously 
$L^1(d\bvp)$ to $L^\infty(d\bvp)$  and $L^2(d\bvp)$ to $L^2(d\bvp)$ we find the following
(cf. \cite[Lemma 2.6, Lemma 3.3]{Sad}).
\begin{lemma}\label{lem-op-T}
$ $
\begin{enumerate}[{\rm (i)}]
\item $\Hh$ is a Hilbert space and $T$ extends to a unitary operator on $\Hh$. It also defines a bounded linear operator from
$\Hh_1$ to $\Hh_\infty$ and from $\Ll^1$ to $\Ll^\infty$.
\item $\Kk$ is a Hilbert space and $\Tt$ is unitary on $\Kk$. It also defines a bounded linear operator from
$\Kk_1$ to $\Kk_\infty$ and $\widehat \Ll^1$ to $\widehat \Ll^\infty$.
\item For any $B\in\Sym_\CC(m),\,\im(B)>0$, $p\in[1,\infty]$ we find that $\PE(B)$ is dense in $\Hh$ and $\Hh_p$.
\item For any $B,C\in\Sym_\CC(m),\,\im(B)>0, \im(C)>0$, $p\in[1,\infty]$ we find that $\PE(B)\otimes\PE(C)$ is dense in $\Kk$ and $\Kk_p$.
\end{enumerate}
\end{lemma}
From \eqref{eq-T-parts} and \cite[equ. (B.26)]{Sad} one obtains
\begin{equation}
 T^2= {\rm id},\quad \Tt^2={\rm id},\quad T\zeta_G = \zeta_{-G^{-1}},\quad
 \Tt \xi_G = \xi_{-G^{-1}}\,.
\end{equation}
We need to get back to $G$ and $|G|^2$ from the functions $\ze_G$ and $\xi_G$. Therefore we define
\begin{equation}\label{eq-def-bfD}
 \mathbf{D}:=\left( \tfrac{(-1)^{mn+j+k}}{2\pi^{mn}} D_{\I,\I}^{n-1} D_{\I\setminus\{k\},\I \setminus\{j\}}\right)_{j,k\in\I}\,
\end{equation}
which is a $m\times m$ matrix of differential operators.
Then for $G\in\Sym_\CC(m),\;\im(G)>0$ one finds
\begin{equation}\label{eq-G-inv}
 G^{-1}= i\,\int \mathbf{D}\,\zeta_G(\bvp\bvp^\top)\,d^{2mn}\,\bvp\,,\quad
 G=-i\,\int \mathbf{D}\,T\,\zeta_G(\bvp\bvp^\top)\,d^{2mn}\,\bvp\;.
\end{equation}
As $\xi_G(M_+,M_-) = \ze_G(M_+)\overline{\ze_G(M_-)}$,  Fubini leads to
\begin{equation}\label{eq-|G|}
 |G|^2=G^*G=\int \mathbf{D}^{(-)} \mathbf{D}^{(+)}\,\Tt\,\xi_G(\bvp_+\bvp_+^\top,\bvp_-\bvp_-^\top)\,d^{2mn}\,\bvp_+ d^{2mn}\,\bvp_-
\end{equation}
where $\mathbf{D}^{(\pm)}$ is the operator $\mathbf{D}$ with respect to the entries $M_\pm=\bvp_\pm \bvp^\top_\pm$.

Equation \eqref{eq-G-inv} can either be obtained using a super symmetric formalism or by realizing that
$D_{\I,\I} \zeta_G = (\frac{i}{2})^m \det(G) \zeta_G$ and $(-1)^{j+k} D_{\I\setminus\{k\},\I\setminus\{j\}} \zeta_G=
(\frac{i}{2})^{m-1} G^{j,k}$, where $G^{j,k}$ is the cofactor of the $j,k$ element of $G$.
Moreover, using the Gaussian integral identity as in \eqref{eq-Gausint}, one has
\begin{equation}
\int  \zeta_G(\bvp\bvp^\top)\,d^{2mn}\bvp \,=\, \frac{(2\pi)^{mn}}{(-i)^{mn} (\det(G))^{n}}\;.
\end{equation}
Combining all these facts with Kramer's rule, $(G^{-1})_{jk} = \frac{G^{j,k}}{\det(G)}$, one obtains the first equation in \eqref{eq-G-inv}. The second one follows from the first one and $T\zeta_G=\zeta_{-G^{-1}}$.

\vspace{.2cm}

Now we can turn back to the analysis of the recursion equations \eqref{eq-zeta-recursion} and \eqref{eq-xi-recursion}.
For simplified notations, let us introduce
\begin{equation}
 \vze_{\lb,z} = \pmat{
  \ze_{\lb,z}^{(0)} \\ \vdots \\ \zeta_{\lb,z}^{(L)} }
  \qtx{and} 
 \vxi_{\lb,z} =
 \begin{pmatrix}
  \xi_{\lb,z}^{(0)} \\ \vdots \\ \xi_{\lb,z}^{(L)}
 \end{pmatrix}\;.
\end{equation}

Using H\"older's inequality, Dominated Convergence and the exponential decay of the functions $\zeta_G$, $\xi_G$ for $\im G>0$ as well as for $\vze_{0,E}$ and $\vxi_{0,E}$ for $E\in I^{K,L}_A$, one obtains the following completely analogue to \cite[Proposition~4.1 and 4.2]{Sad}.

\begin{prop} \label{zeta} We have:
\begin{enumerate}[{\rm (i)}]
\item For $\eta=\im z\geq 0$ the operator $B_{\lb,z}$ is a bounded operator
on $\Hh$ and $\Hh_1$. The map 
\begin{equation}\label{def-F}
F\,:\,
\left(\lambda,E,\eta,f_0,f_1,\ldots,f_L\right) \mapsto\, T B_{\lambda,E+i\eta} 
((f_0^K f_1), f_2, f_3, \ldots , f_L, f_0)
\end{equation}
is a continuous map from 
$\RR\times\RR\times[0,\infty)\times \Hh_\infty\times\Hh^L$ to $\Hh_\infty\times\Hh^L$.
Here $T B_{\lb,z}$ applied to a vector of functions means that we apply it to every
function, $TB_{\lb,z}\left(\wt f_0, \ldots, \wt f_L\right)=
\left(TB_{\lb,z} \wt f_0,\ldots,TB_{\lb,z} \wt f_L\right)$.

Analogously, for $\eta=\im z\geq 0$ the operator $\cB_{\lb,z}$ is a bounded operator on $\Kk$ and $\Kk_1$
and the map
\begin{equation} \label{def-Q}
Q\,:\,(\lambda,E,\eta,g_0,g_1,\ldots,g_L)\mapsto \Tt\cB_{\lambda,E+i\eta} 
((g_0^K g_1), g_2, \ldots , g_L, g_0)
\end{equation} 
is a continuous map from $\RR\times\RR\times [0,\infty)\times \Kk_\infty \times \Kk^L$ to $\Kk_\infty\times\Kk^L$.
\item   $\vze_{\lb,z} \in {\Hh}^{L+1}_\infty \subset \Hh_\infty \times \Hh^L$, $\vxi_{\lb,z} \in \Kk_\infty^{L+1} \subset \Kk_\infty\times \Kk^L$
for all $\lb \in \RR$ and $z=E+i\eta$ with $\eta >0$.   
The maps $(\lb,E, \eta) \to \vze_{\lb,E +i\eta}$ and $(\lb,E, \eta) \to \vxi_{\lb,E +i\eta}$ are continuous from
$\RR \times \RR \times (0,\infty)$ to $ \Hh_\infty \times \Hh^L$ and to $\Kk_\infty \times \Kk^L$, respectively.
\item  If $E\,\in\, I^{K,L}_{A}$, then  $\vze_{0,E} \in (\PE(m))^{L+1} \subset \Hh_\infty\times\Hh^L$,
$\vxi_{0,E} \in \Kk_\infty\times \Kk^L$ and
\begin{equation}
\lim_{\eta \downarrow 0} \vze_{0,E+i\eta}\; = \;\vze_{0,E} \;\;\mbox{in}\;\; \Hh_\infty\times\Hh^L \;,\quad
\lim_{\eta \downarrow 0} \vxi_{0,E+i\eta}\; = \;\vxi_{0,E} \;\;\mbox{in}\;\; \Kk_\infty \times \Kk^L\;.   
\end{equation}

\item The equalities  \eqref{eq-zeta-recursion} and \eqref{eq-xi-recursion} 
can be rewritten as fixed point equations in 
$\Hh\times\Hh_\infty$ and $\Kk\times\Kk_\infty$, respectively, 
\begin{equation} \label{eq-fpz}
\vze_{\lb,z}  \; = \;F(\lb,E,\eta,\vze_{\lb,z})\;,\quad
\vxi_{\lb,z}  \; = \;Q(\lb,E,\eta,\vxi_{\lb,z})\;,
\end{equation}
valid for all $\lb \in \RR$ and $z=E+i\eta$ with $\eta >0$, and also valid for  $\lb=0$ and
 $z=E$ with $E\,\in\, I^{K,L}_A$. 
\end{enumerate}
\end{prop}

\begin{remark} \label{rem-tech-diff}
One of the differences between this work and \cite{Sad} is that for fixed $(\lambda,E,\eta)$
the maps $F$ and $Q$ are operators on $\Hh_\infty\times\Hh^L$ and $\Kk_\infty\times\Kk^L$, respectively, but not on $\Hh_\infty^{L+1}$ or $\Kk_\infty^{L+1}$. 
We can not use the space $\Hh^{L+1}_\infty$ as for $\eta=0$ and $f_p\in\Hh_\infty$ we can not 
show that $B_{\lb,E} f_p$ is in $\Hh_1$ so that $T$ would map it back to
$\Hh_\infty$. We can only say that one lands in $\Hh$ after applying $T$. 
This in turn comes from the fact that there is only one factor
and not a product of more than one function after
$B_{\lb,z}$ in the last $L$ entries
which resembles the fact that vertices of label $1,\ldots,L$ do have only
one child in $\TT^{K,L}$.
Recall that in \cite{Sad} every vertex had to have at least two children.\\
The other difference is that we do not need the smaller spaces 
$\Hh_\infty^{(0)},\,\Kk_\infty^{(0)}$ introduced in \cite{Sad} to
avoid some further assumption (but one could work in the spaces $\Hh^{(0)}_\infty\times(\Hh^{(0)})^L$,
$\Kk^{(0)}_\infty\times(\Kk^{(0)})^L$ if one wanted to).
\end{remark}


\section{Spectrum of Frechet derivatives \label{sec-frechet}}


It is now time to introduce some more notations to properly describe the spectrum 
of the Frechet derivatives. These notations were also used in \cite{Sad}.
By $\Delta(m,\ZZ_+)$ we denote the set of upper triangular matrices with non-negative integer entries.
For $J\in\Delta(m,\ZZ_+)$ and $E\in  I^{K,L}_{A}$ we define
\begin{align}
\theta^{(p)}_{J,E} &:= \prod_{\substack{j,k\in\{1,\ldots,m\} \\j\leq k}} \left[\Gamma^{(p)}_{E-a_j} \Gamma^{(q)}_{E-a_k} \right]^{J_{jk}}\;
\in\,\CC\,,\;p=0,1,\ldots,L\,,\qtx{and} 
\label{eq-def-theta}
\\
\blb_{J,E}&:=\diag(\theta^{(0)}_{J,E},\ldots,\theta^{(L)}_{J,E})=
\pmat{
 \theta^{(0)}_{J,E} &  \\  & \ddots \\ & & \theta^{(L)}_{J,E} }
\label{eq-def-blb}
\end{align}

With the help of these matrices we will express the spectrum of the important Frechet derivatives.
The following lemma corresponds to \cite[Lemma 5.1 and 5.2]{Sad}.

\begin{lemma}\label{lem-DF}
We have:
\begin{enumerate}[{\rm (i)}]
\item The map $F$ as in \eqref{def-F}
is continuous and Frechet-differentiable w.r.t. $\vec{f}=(f_0,f_1,\ldots,f_L)\in\Hh_\infty\times\Hh^L$. 
For $\vec f \in  \Hh_\infty^{L+1}$ the
Frechet derivative $F_{\vec f}$
extends naturally to a bounded operator on $\Hh^{L+1}$
which we will also denote as $F_{\vec{f}}$.\\
Similarly, the map $Q$ is Frechet-differentiable w.r.t. $\vec{g}\in\Kk_\infty \times\Kk^L$. 
The derivative $Q_{\vec g}$ is a bounded linear operator on $\Kk_\infty\times\Kk^L$ and 
for $\vec g \in \Kk^{L+1}_\infty$ it extends naturally to a bounded operator on $\Kk^L$. 

\item For $E\in I^{K,L}_{A}$ let $C_E=F_{\vec{f}}(0,E,0,\vec{\ze}_{0,E})$ and
 $\Cc_E=Q_{\vec{g}}(0,E,0,\vxi_{0,E})$.
Then $C_E^{2(L+1)}$ is a compact operator on $\Hh_\infty \times \Hh^L$ and $\Hh^{L+1}$ and 
 $\Cc_E^{2(L+1)}$ is a compact operator on $\Kk_\infty\times\Kk^L$ and $\Kk^{L+1}$.

\item
The spectrum of $C_E$ as an operator on the Hilbert space  $\Hh^{L+1}$ is given by the eigenvalues 
of the matrices $\blb_{J,E} S^{K,L}$ for $J\in\Delta(m,\ZZ_+)$ and the accumulation point $0$.
Thus, denoting the spectrum of $C_E$ on $\Hh^{L+1}$ by $\sigma_{\Hh}(C_E)$ one obtains
\begin{equation} \label{eq-lbj}
\sigma_{\Hh}(C_E)=\bigcup_{J\in\Delta(m,\ZZ_+)} \sigma(\blb_{J,E} S^{K,L})\;\cup\;\{0\}\;.
\end{equation}
Similarly, denoting the spectrum of $\Cc_E$ on $\Kk^{L+1}$ by $\sigma_{\Kk}(\Cc_E)$ one finds
\begin{equation}\label{eq-spectrum-Cc_E}
 \sigma_{\Kk}(\Cc_E)=\bigcup_{J,J'\in\Delta(m,\ZZ_+)} \sigma(\blb_{J,E}\blb^*_{J',E}S^{K,L}) \cup\{0\}\,.
\end{equation}
Here $\blb_{J,E}$ are the matrices as defined in \eqref{eq-def-blb}.

\item The spectra of  $C_E$ and $\Cc_E$ as operators on 
$\Hh_\infty\times \Hh^L$ and $\Kk_\infty \times\Kk^L$, respectively, (denoted by $\sigma_{\Hh_\infty\times\Hh^L}(C_E)$ and
$\sigma_{\Kk_\infty\times\Kk^L}(\Cc_E)$), are the same as their spectra as operators on $\Hh^{L+1}$ and $\Kk^{L+1}$, respectively,  
\beq \label{eq-spectrum-CEinfty}
\sigma_{\Hh_\infty\times\Hh^L}(C_E) = \sigma_{\Hh}(C_E)\;,\quad
\sigma_{\Kk_\infty\times\Kk^L}(\Cc_E) = \sigma_{\Kk}(\Cc_E)\;.
\eeq
\end{enumerate}
\end{lemma}

\noindent{\bf Proof.} For the proof we will mostly just consider the function $F$ and operator $C_E$. The corresponding statements for $Q$ and $\Cc_E$ are proved analogously.

(i) The derivative $F_{\vec{f}}$ can be written
as a $(L+1)\times(L+1)$ matrix of operators and we obtain formally
\begin{equation}\label{eq-der-F0}
F_{\vec{f}}
=\pmat{  KTB_{\lb,z} \Mm(f_0^{K-1} f_1) & TB_{\lambda,z}\Mm(f_0^K) & \\
  0 & 0 & TB_{\lambda,z} & \\ \vdots & & & \ddots \\ 0 & & & & TB_{\lb,z} \\ TB_{\lb,z} & 0 & 0 & \cdots & 0 }\,
\end{equation}
where $\Mm(f)$ denotes the multiplication operator that multiplies a function with $f$.
For $f_0, g_0 \in \Hh_\infty$ and $f_1, g_1 \in \Hh$ one finds 
$B_{\lb,z} f_0^{K-1} f_1 g_0,\,B_{\lb,z} f_0^K g_1 \in \Hh_1$ 
and by Lemma~\ref{lem-op-T} $F_{\vec f}$ defines
a bounded linear operator on $\Hh_\infty\times \Hh^L$. Thus, $F$ is Frechet-differentiable.
Similarly, if $f_1, f_2 \in \Hh_\infty$, then $F_{\vec{f}}$ defines also a bounded linear operator on $\Hh^{L+1}$.

To get (ii) note that $\vze_{0,E} \in (\PE(m))^2$,
and $B_{0,E} = \Mm(e^{\frac i2\Tr((E-A)M)})$. 
For simplicity, let us simply write
$B$ for $B_{0,E}$, $M_1$ for $\Mm((\ze_{0,E}^{(0)})^{K-1} \ze_{0,E}^{(1)})$ and $M_2$ for $\Mm((\ze_{0,E}^{(0)})^K)$, then
\begin{align}
 C_E &= KTB\pmat{M_1 & \\ & \nul} \,+\,TB
\pmat{ 0 & M_0 \\  \vdots & & 1 \\ 0 & & & \ddots \\
& & & & 1 \\ 1 & 0 & \cdots & & 0} \,.
\end{align}
Now taking $C_E^{L+1}$, each of the occurring non-zero terms has at least one factor
$M_0$ or $M_1$ in it. Therefore, each term occurring in $C_E^{2(L+1)}$ has at least two
terms $M_j$, $j=0,1$ with some $TB$ in between. Considering the structure of $C_E$, these terms
in between $M_j$'s either come from $(C_E)_{00} (C_E)_{0j}=TBM_1TBM_j$, $j=0,1$, or
$(C_E)_{01} (C_E)_{12} (C_E)_{23}\cdots (C_E)_{L-1,L} (C_E)_{L,0} (C_E)_{0,j}=TB M_1 (TB)^{L+1} M_j$ for $j=0,1$. Hence, each appearing term in $C_E^{2(L+1)}$ includes
either a term $M_i TB M_j$ or $M_i (TB)^{L+1} M_j$.
We claim that we can use Proposition~\ref{lem-op-compact}~(ii) to obtain that these operators are compact. As $B=\Mm(\zeta_{E-A})$ this means that
we need to check that the $L$ matrices $\Aa_1,\ldots,\Aa_L$ are invertible, where
$\Aa_j$ is a $mj \times mj$ matrix with a tri-diagonal block structure of $m\times m$ blocks
given by
\begin{equation}\label{eq-def-Aa_j}
\Aa_j=\pmat{A-E & \one \\ \one & \ddots & \ddots \\ & \ddots & \ddots & \one \\ & & \one & A-E}\;,
\qtx{where} \Aa_1=A-E\,.
\end{equation}
For $E\in I^{K,L}_A$ one finds for each eigenvalue $a_j$ of $A$ that $E-a_j \in I^{K,L}$.
In particular, $E-a_j \not \in \Ee_L$ as defined in \eqref{eq-def-EL} which yields invertibility of $\Aa_1,\ldots,\Aa_L$ 
as can be easily seen when assuming that $A$ is diagonal.
Similar considerations can be done for the operator $\Cc_E$ using Proposition~\ref{lem-op-compact}~(iii).

Part (iii) is the exact same calculation as in \cite{Sad}.
For the reader's convenience let us point out the main ideas.
Define the $(L+1)\times (L+1)$ matrix $\boldsymbol{\ze}_{0,E}=\diag(\ze^{(0)}_{0,E}\,,\ldots,\, \ze^{(L)}_{0,E})$,  let $D\in\Sym(m)$ and $\vec{v}\in\RR^{L+1}$ and start with the identity
\begin{equation}
 C_E\, (\zeta_{tD}\, \boldsymbol{\ze}_{0,E} \vec{v})  = 
 \smat{\ze_{B_0} \\ & \ze_{B_1} \\ & & \ddots \\ & & &\ze_{B_L}} S^{K,L} \vec{v} \label{eq-C_E-0}
\end{equation}
where for small $t$, 
\begin{align} 
B_0 &=-\left( E-A+tD+K A_E^{(0)}+A_E^{(1)} \right)^{-1} &=
A^{(0)}_E + \sum_{k=1}^\infty t^k (A_E^{(0)} D)^k A_E^{(0)}\\
B_p &=-\left(E-A+tD+A^{(p+1)}_E\right)^{-1}&=A^{(p)}_E + \sum_{k=1}^\infty t^k (A_E^{(p)} D)^k A_E^{(p)}\,
\end{align}
with $1\leq p \leq L$ and the convention $A_E^{(L+1)}=A_E^{(0)}$.
Here we use $\ze_A \ze_B = \ze_{A+B}$ and $T\ze_A=\ze_{-A^{-1}}$ as well as the recursion
equations satisfied by the free Green's function for $\lb=0$.
This gives 
\begin{equation}
 C_E\, (\zeta_{tD}\, \boldsymbol{\ze}_{0,E}\vec{v})\,=\,\boldsymbol{\ze}_{0,E} \prod_{k=1}^\infty\diag\left(\big[\zeta_{t^k (A_E^{(p)} D)^k A_E^{(p)}}\big]_{p=0,\ldots,L}\right ) S^{K,L}\vec{v}\,.
\end{equation}

A further expansion of the exponential functions $\zeta_{tD}$ and $\zeta_{t^k\ldots}$ in powers of $t$ and varying the matrix $D$ leads to
\begin{align}
& C_E\,([\Tr(DM)]^k \boldsymbol{\ze}_{0,E}(M) \vec{v})\,= \label{eq-C_E-1}\\
& \qquad \boldsymbol{\ze}_{0,E}(M) \diag\left(\big[(\Tr(A_E^{(p)} D A_E^{(p)}M))^k\big]_{p=0,\ldots, L}\,\right)\,S^{K,L} \vec{v} \,+\,
P_{k,D,E}(M)\,\vec{v}\notag
\end{align}
where $P_{k,D,E}(M)$ is a matrix of polynomials of degree less than $k$.
Mapping the map $P(M)=[\Tr(D M)]^k$ to $\hat P(M)=[\Tr(A^{(p)}_E D A_E^{(p)} M)]^k$ defines a linear map on the set of homogeneous polynomials in entries of $M$ of order $k$. Using
the diagonal structure of $A_E^{(p)}$ one realizes that these linear maps (for all $k$) are also represented by
$P_J(M) \mapsto \theta^{(p)}_{J,E} P_J(M)$ for $J\in\Delta(m,\ZZ_+)$, where $P_J(M)=\prod_{j,k} (M_{jk})^{J_{jk}}$. Thus,
\begin{equation} \label{eq-C_E-2}
C_E (P_J\boldsymbol{\ze}_{0,E}\vec{v})\,=\,P_J\boldsymbol{\ze}_{0,E} \blb_{J,E}\, S\,\vec{v}\,+\,\boldsymbol{\ze}_{0,E}\,p_{J,E}\,\vec{v}
\end{equation}
where $p_{J,E}$ is a matrix of polynomials of degree less than the one of $P_J$.

Let $f_{J,k}=P_J \boldsymbol{\zeta}_{E,0} \vec{e}_k$ with $(\vec{e}_k)_{k=0,\ldots,L}$ being the standard basis in $\RR^{L+1}$. Using the functions $f_{J,k}$ ordered in some way with increasing degree of the polynomial $P_J$, 
the operator $C_E$ can be represented as an infinite block triangular matrix, where the $(L+1)\times(L+1)$ blocks along the diagonal are given by the matrices $\blb_{J,E} S^{K,L}$.
Using the fact that $C_E^{2(L+1)}$ is compact and the density of the span of these functions in $\Hh^{L+1}$, \cite[Proposition A.1]{Sad} immediately implies \eqref{eq-lbj}.

For the operator $\Cc_E$ we start with a similar calculation as \eqref{eq-C_E-0}, replacing $\ze_{tD}$ by  $\zeta_{tD} \otimes \overline\zeta_{tD'} (M_+,M_-) = \ze_{tD}(M_+) \overline\ze_{tD'}(M_-)$\
and $\boldsymbol{\ze}_{0,E}$ by $\boldsymbol{\xi}_{0,E}(M_+,M_-)=\boldsymbol{\ze}_{0,E}(M_+)\overline{\boldsymbol{\ze}}_{0,E}(M_-) $,
then one obtains similar to \eqref{eq-C_E-1}
\begin{align}
 & \Cc_E\,([\Tr(DM_+-\overline{D'}M_-)]^k \boldsymbol{\xi}_{0,E}(M_+,M_-) \vec{v})\,=\, \\
 & \quad \boldsymbol{\xi}_{0,E} \diag\left([(g^{(p)}(M_+,M_-))^k\,]_{p=0,\ldots,L} \right)\,S^{K,L}\,\vec{v}\,+\, \widetilde P_{k,D,E}(M_+,M_-)\,\vec{v}\notag
\end{align}
where $g^{(p)}(M_+,M_-)=\Tr(A_E^{(p)}DA_E^{(p)}M_+ - \overline{A_E^{(p)}D'A_E^{(p)}} M_-)$, 
and $\widetilde P_{k,D,E}(M_+,M_-)$ is a matrix of polynomial of degree less than $k$ in the entries of
$M_+$ and $M_-$. Using $P_J\otimes P_{J'} (M_+,M_-)=P_J(M_+) P_{J'}(M_-)$ this leads to
\begin{align}
 \Cc_E\,\left( P_J\otimes P_{J'}\,\boldsymbol{\xi}_{0,E}\,\vec{v}\right)\,=\,P_J\otimes P_{J'}\,\boldsymbol{\xi}_{0,E}\,\blb_{J,E} \blb_{J',E}^*\,S^{K,L}\,\vec{v}\,+\,
 \boldsymbol{\xi}_{0,E}\,\widetilde{p}_{J,J',E}\,\vec{v}
\end{align}
where $\widetilde p_{J,J',E}(M_+,M_-)$ is a matrix of polynomials of degree less than the one of $P_J(M_+)P_{J'}(M_-)=P_J \otimes P_{J'}(M_+,M_-)$.
Now we follow the same line of arguments as for the spectrum of $C_E$.

For (iv) note that
$\sigma_{\Hh_\infty\times\Hh^L}(C_E)\subset 
 \sigma_{\Hh}(C_E)$ by compactness of
$C_E^{2(L+1)}$ in $\Hh_\infty\times \Hh^L\,\subset\,\Hh^{L+1}$. 
Equality follows as one finds eigenfunctions corresponding to the eigenvalues of $\blb_{J,E} S^{K,L}$ in
$\Hh_\infty^{L+1}$ by considering the finite dimensional subspaces $\VV_c$
spanned by $f_{J,k}$ with $\|J\|_1\leq c$ (where $\|J\|_1=\sum_{j,k} |J_{j,k}|$)
that are left invariant by $C_E$.

\hfill $\Box$

\vspace{.2cm}

The following result will ensure that we can use the Implicit Function Theorem.
\begin{lemma}\label{lem-not-1}
For $E\in I^{K,L}_A$ and any $J,J' \in \Delta(m,\ZZ_+)$ we find
\begin{equation}
\det(\one-\blb_{J,E} \blb_{J',E}^* S^{K,L}) \neq 0.
\end{equation}
This means, the matrices $\blb_{J,E} \blb_{J',E}^* S$ 
do not have an eigenvalue 1.
In particular, noting $\blb_{\nul,E}=\one$, this implies
\beq \label{eq-spectrum-CE}
1 \notin \sigma_{\Hh}(C_E) \qtx{and}
1 \notin \sigma_{\Kk}(\Cc_E)\;.
\eeq
\end{lemma}
\begin{proof}
For $J,J' \in\Delta(m,\ZZ_+)$ and $E\in I^{K,L}_A$ define
\begin{align}
f(J,J',E)=
\det (\one-\blb_{J,E} \blb_{J',E}^* S^{K,L}) = \notag
1-K \theta^{(0)}_{J,E} \left(\theta^{(0)}_{J',E}\right)^*-
\prod_{p=0}^L \theta^{(p)}_{J,E} \left(\theta^{(p)}_{J',E}\right)^*\;.
\end{align}
For $J=J'=0$ we have $f(\nul,\nul,E)=\det (\one-S^{K,L})=-K \neq 0$.
Now let $\|J\|_1$ denote the norm given by the sum of the absolute values of all entries of $J$.
Next, we consider the case $\|J\|_1+\|J'\|_1 = 1$, i.e. one of these matrices is zero and the other has one entry.
Both cases are completely analogous so let us just consider $J'=0$, $\|J\|_1=1$.
Then by \eqref{eq-def-theta} one has
\begin{equation}\label{eq-J1-1}
 \theta^{(q)}_{J,E}=\Gamma^{(q)}_{E-b_1} \Gamma^{(q)}_{E-b_2}
\end{equation}
for some $b_1, b_2 \in\{a_1,\ldots, a_m\}$.
Using \eqref{eq-gamma-rels} in Proposition~\ref{prop-Gamma} and the Cauchy-Schwarz inequality we find
\begin{align}\label{eq-J1-2}
 K \left|\theta^{(0)}_{J,E}\right|+\prod_{p=0}^L \left|\theta^{(p)}_{J,E}\right| &= 
 K \left|\Gamma^{(0)}_{E-b_1}\right|\left|\Gamma^{(0)}_{E-b_2}\right|+
 \left(\prod_{p=0}^L \left|\Gamma^{(p)}_{E-b_1}\right|\right)
 \left(\prod_{p=0}^L\left|\Gamma^{(p)}_{E-b_2}\right|\right)
 \notag  \\
 &\leq
 \prod_{k=1}^2 \sqrt{K \left|\Gamma^{(0)}_{E-b_k}\right|^2+
 \prod_{p=0}^L \left|\Gamma^{(p)}_{E-b_k}\right|^2} 
 = 1\;.
\end{align}
Since $\theta^{(0)}_{J,E}$ 
is the product of two factors with positive imaginary part, it can not be a positive real number,
hence
$| f(J,\nul,E)|>1-|\theta^{(0)}_{J,E}|-\prod_{p=0}^L|\theta^{(p)}_{J,E}| \geq 0$, 
so $f(J,\nul,E)$ can not be zero in this case. 

Finally, consider $\|J\|_1 + \|J'\|_1 \geq 2$. We may assume without loss of generality that $J\neq \nul$.
Then
\begin{equation}
\theta^{(p)}_{J,E} \left(\theta^{(p)}_{J',E}\right)^* =
\Gamma^{(p)}_{E-b_1} \Gamma^{(p)}_{E-b_2} \cdot\,X^{(p)}
\end{equation}
where $X^{(p)}$ itself is a product of an even number of factors (at least 2) $\Gamma^{(p)}_{E-b}$ or complex conjugates.
By \eqref{eq-gamma-rels} and the fact that for $E\in I^{K,L}_A$ and  $b\in \{a_1,\ldots,a_m\}$ none of the 
imaginary parts of $\Gamma^{(p)}_{E-b}$ can be zero, we find
$|\Gamma^{(0)}_{E-b}|<1$ and $\prod_{p=0}^L|\Gamma^{(p)}_{E-b}|<1$ leading to
$|X^{(0)}|<1$ and $\prod_{p=0}^L|X^{(p)}|< 1$.
Using this and Cauchy-Schwartz as in \eqref{eq-J1-1} we find
\begin{equation}
 K \left|\theta^{(0)}_{J,E} \left(\theta^{(0)}_{J',E}\right)^*\right| + 
 \prod_{p=0}^L\left| \theta^{(p)}_{J,E} \left(\theta^{(p)}_{J',E}\right)^* \right|
 \,<\,1
\end{equation}
which immediately implies $|f(J,J',E)|>0$.
Hence, in any case, $f(J,J',E)$ will not be zero.
\end{proof}


   
\section{Conclusions \label{sec-proofs}}


The most important ingredient for the proof of Theorem~\ref{th:main2} is the following.

\begin{prop}  \label{prop-extend-1}
There exists an open set $U_1\subset \RR^2$ with $\{0\}\times I^{K,L}_A \subset U_1$,
 such that the maps 
\begin{align}
&(\lb, E,\eta) \in U_1 \times (0, \infty) \;\mapsto\;
 \vze_{\lb,E+i\eta} \,\in \Hh_\infty\times\Hh^L  \label{xizeze} \\
&(\lb, E,\eta) \in U_1 \times (0, \infty) \;\mapsto\;
 \vxi_{\lb,E+i\eta} \,\in \Kk_\infty\times \Kk^L  \label{xizexi}
\end{align}
 have  continuous extensions to maps from $ U_1 \times [0, \infty)$ to $\Hh_\infty\times\Hh^L$
 and $\Kk_\infty\times \Kk^L$, respectively,  that
 satisfy \eqref{eq-fpz}. 
\end{prop}

\begin{proof}
By Lemma~\ref{lem-DF} and Lemma~\ref{lem-not-1} we can use the Implicit Function Theorem on Banach Spaces 
as stated in \cite[Appendix~B]{Kl6} for the functions
$\hat F(\lb,E,\eta,\vec f)= F(\lb, E,\eta,\vec f)-\vec f$ and
$\hat Q(\lb,E,\eta,\vec g)= Q(\lb,E,\eta,\vec g)-\vec g$ at the points $(0,E,0,\vze_{0,E})$
and $(0,E,0,\vxi_{0,E})$ with $E\in I^{K,L}_A$.
Uniqueness of the continuous implicit function and the continuity properties of
$\vze_{\lb,E+i\eta} $ and $\vxi_{\lb,E+i\eta}$ as stated in
Proposition~\ref{zeta} give the continuous extensions.
\end{proof}

For $\eta=\im(z)>0$ let us define general averaged quantities $\ze_{\lb,z}^{(x|y)}=\E\, \ze_{G_{\lb,z}^{(x|y)}}$ and 
$\xi_{\lb,z}^{(x|y)}=\E\, \xi_{G_{\lb,z}^{(x|y)}}$ where as in Section~\ref{sec-rec} the upper index $(x|y)$ for $x,y \in \TT^{K,L}_p$ indicates that we consider the Green's function at $x$
for the operator $H_\lb^{(x|y)}$ which is the restriction of $H_\lb$ to the subtree  $(\TT^{K,L}_p)^{(x|y)}$ that is obtained by removing the branch at $x$ going through $y$. 
If $x$ is a child or descendant of $y$ then $\ze_{\lb,z}^{(x|y)}=\ze_{\lb,z}^{(l(x))}$ and $\xi_{\lb,z}^{(x|y)}=\xi_{\lb,z}^{(l(x))}$, where $l(x)$ is the label of $x$.
But if $y$ is a descendant of $x$ then we get different quantities.
In the following arguments it will often be used implicitly that $B_{\lb,z}$ and
$\Bb_{\lb,z}$ is a strongly continuous family of operators on any space $\Ll^r, \,\Hh_r$ and
$\widehat \Ll^r,\,\Kk_r$, respectively, for $r\in[1,\infty)$ which follows from the Leibniz rule \eqref{eq-Leibn0}, boundedness and Dominated Convergence.

\begin{prop}\label{prop-extend-2}
 There is an open set $U_2\subset \RR^2$, $\{0\}\times I^{K,L}_A\subset U_2$, such that for all 
 $p\in\{1,\ldots,L\}$, $x\in\TT^{K,L}_p$ with $d(0^{(p)},x)\leq L-p$ and $y$ being the unique child of $x$, one has that the maps
 \begin{equation}\label{eq-pointw-conv}
  (\lambda,E,\eta,M) \mapsto \zeta_{\lb,E+i\eta}^{(x|y)}(M) \qtx{and} 
  (\lambda,E,\eta,M_+,M_-) \mapsto \xi_{\lb,E+i\eta}^{(x|y)}(M_+,M_-)
 \end{equation}
 extend continuously as maps from $U_2\times[0,\infty)\times \Sym^+(m)$ and
 $U_2\times[0,\infty)\times (\Sym^+(m))^2$ to $\CC$, respectively.
 Moreover, $\hn \ze_{\lb,z}^{(x|y)} \hn_\infty$ and
 $\hnn\xi_{\lb,z}^{(x|y)}\hnn_\infty$ are uniformly bounded 
 on compact subsets of $U\times[0,\infty)$.
\end{prop}

\begin{proof}
 First note that $d(0^{(p)},x)\leq L-p$ means that $x$ is in the starting line segment of $\TT^{K,L}_p$ and hence $(\TT^{K,L}_p)^{(x|y)}$ is a finite line with $d(0^{(p)},x)$ edges and
 $j=d(0^{(p)},x)+1\leq L$ vertices. In fact, $H^{(x|y)}_{0}-E$ is given by the matrix $\Aa_j$
 as in \eqref{eq-def-Aa_j} and using $E\in I^{K,L}_A$ which implies 
 $E-a_j\not\in\Ee_L$ for any eigenvalue of $a_j$ one obtains that
 $\zeta_{0,E}^{(x|y)}$ and $\xi_{0,E}^{(x|y)}$ exist.
 Using assumption (V), the boundedness of the distribution of $V(x)$, one obtains
 existence of $\zeta_{\lb,E}^{(x|y)}$ and $\xi_{\lb,E}^{(x|y)}$ 
 for $(\lb,E)$ in an open neighborhood of $\{0\}\times I^{K,L}_A$.
  Point wise continuity of the maps follows immediately.
 Note that the infinity norm of $\ze_G$ and $\xi_G$ are bounded by $1$ and the derivatives appearing in the $\hn \cdot \hn_\infty$ and $\hnn \cdot \hnn_\infty$ norms lead to
 multiplication by determinants of minors of $G$. Therefore, using assumption (V) again
 we obtain the uniform bounds of the $\hn\cdot\hn_\infty$ and $\hnn\cdot\hnn_\infty$ norm on compact subsets of $U\times[0,\infty)$.
\end{proof}

\begin{prop}\label{prop-xize}
Let $U=U_1\cap U_2$ with $U_1$ and $U_2$ as in Propositions~\ref{prop-extend-1} and \ref{prop-extend-2}. Clearly, $U\subset\RR^2$ is open and $\{0\} \times I^{K,L}_A \subset U$.
For all $x\in\TT^{K,L}_0$ and all $x\in\TT^{K,L}_p$ with $d(x,0^{(p)})>L-p$
and all children $y$ of $x$ there is $r(x,y)\in\{2,\infty\}$ such that the maps
\begin{align}
 & (\lb,E,\eta)\in\RR\times\RR\times(0,\infty)\;\mapsto\;\ze^{(x|y)}_{\lb,E+i\eta} \in \Ll^{r(x,y)} \;,
\label{eq-cext1}\\
& (\lb,E,\eta)\in\RR\times\RR\times(0,\infty)\;\mapsto\;\xi^{(x|y)}_{\lb,E+i\eta} \in 
\widehat{\Ll}^{r(x,y)}\;, \label{eq-cext2}
\end{align}
have continuous extensions to maps from $(\lb,E,\eta)\in U\times[0,\infty)$ to
$\Ll^{r(x,y)}$ and $\widehat\Ll^{r(x,y)}$, respectively.
Moreover, for $l(x)=0$ and $l(y)=1$ as well as for $l(x)\neq 0$
we have $r(x,y)=2$ and hence $\Ll^{r(x,y)}=\Ll^2=\Hh$.
For $l(x)=0=l(y)$ both, $r(x,y)=2$ and $r(x,y)=\infty$ are possible.
Here, $\Ll^r$ and $\widehat \Ll^r$ denote the spaces as defined in Definition~\ref{def-spaces}.
\end{prop}
For such continuous extensions of maps from $(\lb,E+i\eta)$  that extend as functions from 
$U\times[0,\infty)$ to $\Ll^r$ we will use the notion that such a family of functions extends continuously in $\Ll^r$.
\begin{proof}
All arguments will implicitly use some specific version of the Green's matrix recursion \eqref{eq-G-rec1} in the form as in \eqref{eq-zeta-recursion}.
We will also implicitly use Proposition~\ref{prop-extend-1} and H\"older's inequalities.

Note that $\Hh_\infty\subset\Ll^2 \cap \Ll^\infty$, thus $\ze_{\lb,z}^{(0)}$ extends continuously in $\Ll^2$ and $\Ll^\infty$, where $\ze_{\lb,z}^{(p)}$ extends continuously in $\Ll^2=\Hh$.
The proof will be done by induction over the distance from the root.
For the start on $\TT^{K,L}_0$ we have to consider the root $0^{(0)}$ and on
$\TT^{K,L}_p$ we have to consider the vertex $1_0^{(p)}$ as defined in Theorem~\ref{th:main} which is the closest vertex to the root of label $0$ and characterized by
$d(0^{(p)},1_0^{(p)})=L+1-p$.
Let $y$ be a child of $0^{(0)}$, then using the general recursion relation
\eqref{eq-G-rec1} in the form as \eqref{eq-G-rec3} and taking expectations leads to
$$
\ze^{(0^{(0)}|y)}_{\lb,z}=TB_{\lb,z} \big( (\ze^{(0)}_{\lb,z})^{K-1} \ze^{(l(y))}_{\lb,z} \big)\,
$$
which by Proposition~\ref{prop-extend-1} gives the continuous extension in $\Hh=\Ll^{2}$ 
(even $\Hh_\infty$ if $K\geq 2$), thus $r(0^{(0)},y)=2$.
Similar, letting $x$ be the parent of $1_0^{(p)}$ and $y$ a child, then
$$
\ze^{(1_0^{(p)}|y)}= T B_{\lb,z}\big( \ze_{\lb,z}^{(x|1_0^{(p)})}  
(\ze_{\lb,z}^{(0)})^{K-1} \ze_{\lb,z}^{(l(y))} \big)
$$
Using Propositions~\ref{prop-extend-1} and \ref{prop-extend-2} and Dominated Convergence one obtains that the product after the operators $T B_{\lb,z}$ on the right hand side extend continuously in $\Ll^2$ and hence the whole term does too. 
In particular, $r(1_0^{(p)},y)=2$.

For the induction step, let $x$ be a descendant of $0^{(0)}$
or $1_0^{(p)}$ for some $p\in\{1,\ldots,L\}$.
Let $x_0$ be the parent and $y$ some child of $x$. We have several cases:\\
Case 1: $l(x) = q \neq 1$, then $l(x_0)=q-1$ and $l(y)=q+1$ or $l(y)=0$ if $q=L$.
We have by induction assumption that $r(x_0,x)=2$ and so $\ze^{(x_0|x)}_{\lb,z}$ extends continuously in
$\Ll^2$. Hence, $\ze_{\lb,z}^{(x|y)}=TB_{\lb,z} \ze_{\lb,z}^{(x_0|x)}$ does as well
and $r(x,y)=2$.\\
Case 2: $l(x)=0$ and $l(y)=1$, then $\ze_{\lb,z}^{(x|y)}=TB_{\lb,z} \big( 
(\ze^{(0)}_{\lb,z})^K \ze^{(x_0|x)}_{\lb,z} \big)$.
By induction assumption $\ze_{\lb,z}^{(x_0|x)}$ either
extends in $\Ll^2$ or $\Ll^\infty$. As $K\geq 1$ we get an extension in $\Ll^2$ in either case, so $r(x,y)=2$.\\
Case 3: $l(x)=0$ and $l(y)=0$, then
$\ze_{\lb,z}^{(x|y)} = TB_{\lb,z} \big( (\ze_{\lb,z}^{(0)})^{K-1} 
\ze_{\lb,z}^{(1)} \ze_{\lb,z}^{(x_0|x)}\big)$.
If $\ze_{\lb,z}^{(x_0|x)}$ extends continuously in $\Ll^2$, then the product after $TB_{\lb,z}$ extends continuously in $\Ll^1$ and hence $\ze_{\lb,z}^{(x|y)}$ extends continuously in $\Ll^\infty$. If $\ze_{\lb,z}^{(x_0|x)}$ extends continuously in $\Ll^\infty$ then we
obtain a continuous extension of $\ze_{\lb,z}^{(x|y)}$ in $\Ll^2$.

All arguments for the functions $\xi_{\lb,z}^{(x|y)}$ are completely analogue.
\end{proof}

\begin{proof}[Proof of Theorem~\ref{th:main2}]
Using \eqref{eq-G-inv}, the recursion relation \eqref{eq-G-rec1} and $T^2={\rm id}$
one obtains
\begin{align}
\E\big(G^{[x]}_{\lb}\, (z)\big) &= -i \int \mathbf{D}\,T\,\E\, \ze_{G^{[x]}_\lb(z)}(\bvp\bvp^\top)\,d^{2mn}\bvp \notag \\
\label{eq-EG-x-D}
 &= -i\int {\mathbf D}B_{\lb,z} \prod_{y:d(x,y)=1} \ze^{(y|x)}_{\lb,z}(\bvp\bvp^\top)\,
d^{2mn} \bvp\;
\end{align}
and similarly, based on \eqref{eq-|G|} one obtains
\begin{align}
\E\left(\left| G^{[x]}_{\lb}\, (z)\right|^2\right) &= 
\int {\mathbf D}^{(-)} {\mathbf D}^{(+)} \Bb_{\lb,z} \!\!\!\prod_{y:d(x,y)=1}\!\!\! \xi^{(y|x)}_{\lb,z}(\bvp_+\bvp_+^\top, \bvp_-\bvp_-^\top)
\,d^{2mn} \bvp_+\,d^{2mn} \bvp_- \;, \label{eq-EGG-x-D}
\end{align}
where ${\mathbf D}$ is defined by \eqref{eq-def-bfD},
${\mathbf D}^{(\pm)}$ represent the matrix-operator ${\mathbf D}$ acting with respect to $M_\pm=\bvp_\pm\bvp_\pm^\top$ and
${\mathbf D}^{(-)}{\mathbf D}^{(+)}$ has to be understood as a matrix product.

Using Propositions~\ref{prop-extend-1} and \ref{prop-xize} one obtains for $x\in \TT^{K,L}_0$ or
$d(x,0^{(p)})>L-p$ that the products of the $\ze$'s on the right hand side of \eqref{eq-EG-x-D} 
have 2 factors that extend continuously in $\Ll^2$ and if $l(x)=0$ some additional bunch of factors that extend continuously in $\Ll^\infty$. Therefore, the product extends continuously in $\Ll^1$. Hence, when applying $\mathbf{D}$, each entry of the matrix
extends continuously in $L^1(d^{2mn}\bvp)$.
Therefore, the map $(\lb,E,\eta) \mapsto \E(G^{[x]}_{\lb}(z))\in \Sym_\CC(m)$ extends continuously to a map from $U\times[0,\infty)$ to $\Sym_\CC(m)$ with $U$ as in Proposition~\ref{prop-xize}.
By similar arguments the same is true for $(\lb,E,\eta) \mapsto \E(|G^{[x]}_{\lb}(z)|^2)\in
\Sym^+(m)$. This proves Theorem~\ref{th:main2}.
\end{proof}

\appendix

\section{An identity for the unperturbed Green's functions on trees of finite cone type\label{app-gamma}}

Recall that associated to an $s\times s$ substitution matrix $S\in{\rm Mat}(s,\ZZ_+)$ with non-negative integer entries are
the following $s$ rooted trees of finite cone type, denoted by $\TT_r, r=1,\ldots,s$.
Each vertex has a label, the root of the tree $\TT_r$ has label $r$, any vertex of label $p$ has $S_{pq}$ children of label $q$.
Denoting by $\Delta$ the adjacency operator on the forest $\bigcup_r \TT^{(r)}$ and by $0^{(r)}\in\TT_r$ the root of the
tree $\TT_r$, we define for $\im(z)>0$ the Green's functions
$$
\Gamma^{(r)}_{z} := \langle 0^{(r)}\,|\, (\Delta-z)^{-1}\,|\,0^{(r)} \rangle\;.
$$
We define the set
$$
\Sigma=\{E\in\RR\,:\,\Gamma^{(r)}_E:=\lim_{\eta\downarrow 0} \Gamma^{(r)}_{E+i\eta}\,\text{exists for all $r$ and}
\im(\Gamma^{(q)}_E) > 0\; \text{for some $q$}\}
$$

\begin{prop}\label{prop-gr}
Let
$\Gamma_E=\diag(\Gamma_E^{(1)},\ldots,\Gamma_E^{(s)})$ denote the diagonal $s\times s$ matrix with the
Green's functions along the diagonal.
Then one has for $E\in\Sigma$
\begin{equation}\label{prop-gr-identity}
 \det (\one-|\Gamma_E|^2 S ) \,=\, 0\;.
\end{equation}
\end{prop}

\begin{proof}
 The recursion relation for the Green's functions is given by
 $$
 \Gamma_E^{(p)} = - \left(E+\sum_{q=1}^s S_{pq} \Gamma^{(q)}_E \right)^{-1}\;.
 $$
 Multiplying by $(\Gamma_E^{(p)})^*$ and some algebra leads to
 $$
 \left|\Gamma_E^{(p)}\right|^2 \sum_{q=1}^s S_{pq} \Gamma^{(q)}_E = -E \left|\Gamma_E^{(p)} \right|^2 
 - \left(\Gamma^{(p)}_E\right)^*\,.
 $$
 Taking imaginary parts gives
 $$
 \left|\Gamma_E^{(p)}\right|^2 \sum_{q=1}^s S_{pq} \im(\Gamma^{(q)}_E) = \im(\Gamma^{(p)}_E)
 $$
 Defining the vector $\vec\Gamma_E=(\Gamma^{(1)}_E,\ldots,\Gamma^{(s)}_E)^\top$ these equations can be read as
 $$
 |\Gamma_E|^2 S \,\im(\vec\Gamma_E) = \im(\vec\Gamma_E)\;
 $$
 and for $E\in\Sigma$, $\im(\vec\Gamma)$ is not the zero vector. Hence, $|\Gamma_E|^2 S$ has an eigenvalue $1$ which proves \eqref{prop-gr-identity}.
\end{proof}


\section{Gaussian integrals and the Fourier transform \label{app-FMF}}

The following identities are used at various parts in the article.

\begin{lemma}\label{lem-Gausint}
Let $D$ be an invertible, symmetric $k\times k$ matrix with positive definite real part,
i.e. $D=D^\top,\, \re(D)> 0$.
Then, for any complex vector $v\in\CC^k$ one has the Gaussian integral
\begin{equation}
\label{eq-Gausint}
\int_{\RR^k} e^{-\frac12 (x+v)\cdot D (x+v)}\,d^k x\,=\,\frac{(2\pi)^{k/2}}{\sqrt{\det (D)}} \;.
\end{equation}
Some care needs to be taken to select the correct branch of $\sqrt{\det(D)}$.
If $D=A+iB$ where $A>0$ is the real part, then we write
$D=\sqrt{A} (1+iA^{-1/2} B A^{-1/2}) \sqrt{A}$ 
where $\sqrt{A}$ has the same eigenspaces as $A$ and the corresponding eigenvalues are the 
positive square roots of the eigenvalues of $A$. Furthermore, 
$A^{-1/2} B A^{-1/2}$ is diagonalizable by a real orthogonal matrix. This diagonalizes
$1+iA^{-1/2} B A^{-1/2}$ as well and the eigenvalues
have all real part $1$.
Hence, we may define $\sqrt{1+iA^{-1/2} B A^{-1/2}}$ by taking the same eigenspaces and the 
principal branch of the square roots of the eigenvalues.
Then \eqref{eq-Gausint} is correct with 
$\sqrt{\det(D)}=\det(\sqrt{A}) \det(\sqrt{1+iA^{-1/2} B A^{-1/2}})$.
\end{lemma}

\begin{proof}
In one dimension one has the well known integral formula
\begin{equation}\label{eq-Gaus1d}
\int_{-\infty}^\infty e^{-z(x+c)^2} \,dx = \frac{\sqrt{\pi}}{\sqrt{z}}
\end{equation}
for $\re(z)>0$, where the square root is the principal branch and $c$ is any fixed complex number.
Now if $D=A+iB$, then use a basis change
$y = O \sqrt{A}x$, where $O$ is a real orthogonal matrix such that
$O A^{-1/2}BA^{-1/2} O^\top$ is diagonal. This leads to a Gaussian integral with a diagonal matrix
and then \eqref{eq-Gausint} follows from \eqref{eq-Gaus1d}.
\end{proof}

For functions $f(x)$ on $\RR^k$ and a $k\times k$ matrix $D$ we define
$\Mm(D)$ and $\Cc(D)$ to be the multiplication and convolution operator by $e^{\frac12 i x\cdot Dx}$, i.e.
$$
(\Mm(D) f) (x) = e^{\frac12 i x\cdot Dx} f(x)\;,\quad
(\Cc(D) f)(x) = \int e^{\frac12 i (x-y)\cdot D(x-y)} f(y)\,d^k y\;.
$$
For $D$ invertible we also define $\Ss(D)$ by
$(\Ss(D) f)(x)=f(Dx)$ which is a change of variables and defines a bounded operator on any $L^p$ space.

\begin{lemma}\label{lem-FMF}
Let $\Ff$ denote the Fourier transform on $\RR^k$, and let $D$ be a symmetric, 
invertible $k\times k$ matrix with positive semi-definite
imaginary part $\im(D)\geq 0$.

\begin{enumerate}[{\rm (i)}]
\item Then as a  map from $L^1(\RR^k)\cap L^{2}(\RR^k)$ to $L^2(\RR^k)$ one has
\begin{equation}\label{eq-FMF}
\Ff^*\Mm(D)\Ff
\,=\, 
\frac{\Cc(-D^{-1})} {(2\pi)^{k/2}\sqrt{\det(-iD)}}\;
\end{equation}
where $\sqrt{\det(-iD)}$ is selected as in Lemma~\ref{lem-Gausint}
(note that $\re(-iD)\geq 0$).

\item If $D$ is real, i.e. $\im(D)=0$, this can be re-written as
\begin{equation}\label{eq-FMF2}
 \Ff^* \Mm(D) \Ff\,=\,
 \frac{1}{\sqrt{\det(-iD)}}\,
 \Mm(-D^{-1})\,\Ss(D^{-1}) \,\Ff\,\Mm(-D^{-1})
\end{equation}
Equation \eqref{eq-FMF2} is valid in operator sense on $L^2(\RR^k)$.

\item For a real invertible, symmetric matrix $D$ define $D_1:=-D^{-1}$ and iteratively
define $D_j:=-(D+D_{j-1})^{-1}$ as long as the inverses exist, 
i.e. $D_2=-(D-D^{-1})^{-1}$, $D_3=-(D-(D-D^{-1})^{-1})^{-1}$, and so on.
Assume that the first $L$ matrices $D_1,\ldots, D_{L}$, exist. Then, one has 
as operators on $L^2(\RR^k)$
\begin{equation}\label{eq-FMF3}
\big(\Ff \Mm(D)\,\big)^{L+1}\,=\,
\left[\prod_{j=1}^{L} \frac{\Mm(D_j)\Ss(D_j)}{\sqrt{\det(i D^{-1}_j)}}\right]\,\Ff\,\Mm(D+D_{L})
\end{equation}
Note that all $D_j$ are invertible and
therefore these are indeed bounded operators.
\end{enumerate}
\end{lemma}

\begin{proof}
First assume $\im(D)>0$, then for $f\in L^1(\RR^k)$ and any $y\in \RR^k$, 
the map $(x,w)\mapsto e^{-iy\cdot x}e^{\frac12 i x \cdot D x} e^{ix\cdot w} f(w)$ is in $L^1(\RR^{2k})$ and
one finds
\begin{align}
 & (2\pi)^k\Ff^* \Mm(D) \Ff f)(y)=
 \int e^{-iy\cdot x}e^{\frac12 i x \cdot D x} e^{ix\cdot w} f(w)\,d^kw \;d^kx \notag \\
 & \quad =\;\int\left[\int e^{\frac12[x + D^{-1} (w-y)]\,\cdot\,iD [ x + D^{-1}(w-y)]}\,d^kx\right]\;
 e^{-\frac{i}2 (y-w) \cdot D^{-1} (y-w)}\,f(w)\,d^kw \notag \\
 & \qquad = \; \frac{(2\pi)^{k/2}}{\sqrt{\det(-iD)}} \;\;\Cc(-D^{-1})\, f\, (y)\;.
\end{align}
Now, if $\im(D)$ is only positive semi-definite, we approach $D$ by $D+i\epsilon$ and let 
$f\in L^1(\RR^k)\cap L^2(\RR^k)$. Then as $\epsilon\downarrow 0$, the right hand side converges point wise
(for fixed $y$). As the $L^2$ norm is uniformly bounded by $\|f\|_2$, Dominated Convergence shows convergence in
$L^2(\RR^k)$. As the operators $\Mm(e^{\frac12 i x \cdot (D+i\epsilon) x})$ converge
for $\epsilon\downarrow 0$ in the strong operator topology, we also get convergence on the left hand side in $L^2(\RR^k)$.
For part (ii) and \eqref{eq-FMF2} note that
$$
\int e^{-\frac{i}2 (y-w) \cdot D^{-1} (y-w)}\,f(w)\,d^kw\,=\,
e^{-\frac{i}2 y \cdot D^{-1} y}\int e^{i D^{-1}y\cdot w}
e^{-\frac{i}2 w \cdot D^{-1} w}f(w) d^k w\;.
$$
As the left hand side and right hand side of \eqref{eq-FMF2} are compositions of
bounded operators on $L^2(\RR^k)$ the validity for functions in the dense subset $L^1(\RR^k)\cap L^2(\RR^k)$ implies the validity on $L^2(\RR^k)$.

For (iii) note that \eqref{eq-FMF2} also implies
$$
\Ff \Mm(D) \Ff = \frac{\Mm(-D^{-1}) \Ss(-D^{-1})}{\sqrt{\det(-iD)}}  \Ff \Mm(-D^{-1})=
\frac{\Mm(D_1) \Ss(D_1)}{\sqrt{\det(iD_1^{-1})}}  \Ff \Mm(D_1)\;.
$$
Now iteration and using $\Mm(A) \Mm(B)=\Mm(A+B)$ yields \eqref{eq-FMF3}.
\end{proof}

\begin{remark} \label{rem-FMF}
 For part (iii) one can reformulate the condition that the iteratively defined 
 $k\times k$ matrices $D_j$ exist for $j=1,\ldots,L$. 
 Note, $D_2^{-1}=-D+D^{-1}$ is the Schur complement w.r.t. the first upper block of
 the block matrix $\smat{- D & \one \\ \one & -D}$. Inductively, one obtains that
 $D_j$ is the inverse of the Schur complement of the upper left $k\times k$ block of a $jk\times jk$ matrix $\Dd_j$, This matrix has a tri-diagonal block structure given by $k\times k$ blocks with $-D$ along the diagonal and identity matrices $\one$ on the side diagonals, i.e.
 \begin{equation} \label{eq-Dd-j}
 \Dd_j:=\pmat{- D & \one &  \\ \one & \ddots & \ddots   \\  & \ddots & \ddots & \one \\ & & \one & -D}\,, \qtx{where} \Dd_1 = -D\;.
 \end{equation}
 Now, for a matrix $Y=\smat{U&V\\ W& X}$ with $X$ being invertible one has that the invertibility of $Y$ and the invertibility of the Schur complement $U-VX^{-1}W$ are equivalent.
 Therefore, one obtains by induction that the existence of all the matrices $D_1,\ldots, D_{L}$ is equivalent to the invertibility of all the matrices $\Dd_1,\ldots, \Dd_{L}$.
\end{remark}

\section{Compact operators involving Fourier transforms} \label{app-compact}

In the analysis of the Frechet derivative it is important that s certain 
power is a compact operator. In this work we need a little bit more general results compared to previous work as \cite{Sad} to prove that.
As above, for functions $f(M),\,g(M_+,M_-)$ and $h(x)$, $\Mm(f(M)),\,\Mm(g(M_+,M_-))$ and $\Mm(h(x))$ will denote the corresponding multiplication
operators. Recall $\zeta_B(M)=e^{\frac i2 \Tr(BM)}$, \, 
$\xi_B(M_+,M_-)=e^{\frac i2 \Tr(BM_+-\overline{B} M_-)} $.

For a real, symmetric $k\times k$ matrix $D$ define the $jk\times jk$ matrix $\Dd_j=\Dd_j(D)$ 
as a tri-diagonal $k\times k$ block matrix with $-D$ along the diagonal and the unit matrix $\one$ along the side diagonal, as in \eqref{eq-Dd-j}.
We denote the set of real, symmetric $k\times k$ matrices $D$, where $\Dd_1, \Dd_2\,\ldots, \Dd_L$ are invertible, by $\Ss(k,L)$.

\begin{prop}\label{lem-op-compact}
$ $
\begin{enumerate}[{\rm (i)}]
\item Let $h_1,\, h_2$ be exponentially decaying, continuous functions on $\RR^k$,
let $\Ff$ denote the Fourier transform on $\RR^k$ and let $D\in \Ss(k,L)$.
Then,
\begin{equation} \label{op-base}
 \Mm(h_1) \Ff \Mm(h_2)\,\qtx{and} \Mm(h_1) \left( \Ff \Mm(e^{\frac12 i x\cdot Dx}) \right)^{L+1} \Mm(h_2)
\end{equation}
are compact operators from $L^2(\RR^k)$ to $L^p(\RR^k)$ for any $p\in [1,\infty]$.
\item For $f_1,f_2\,\in \PE(m)$ and $B\in\Ss(m,L)$ the operators
\begin{equation}\label{op-Hh}
\Mm(f_1) T \Mm(f_2)\,,\qtx{and}  \Mm(f_1) \left(T \Mm(\zeta_B)\right)^{L+1}  \Mm(f_2)
\end{equation}
are compact operators from $\Hh_p$ to $\Hh_q$ for any $p,q\in[1,\infty]$.
(Note that the case where one Banach space is $\Hh$ is included as $\Hh=\Hh_2$ as a set, only the norm differs 
technically by a factor $2$).
\item 
For $g_1,g_2\in\PE(m)\otimes \PE(m)$, $B\in\Ss(m,L)$, the operators
\begin{equation}\label{op-Kk}
\Mm(g_1) \Tt \Mm(g_2)\,,\qtx{and}  \Mm(g_1) \left(\Tt \Mm(\xi_B) \right)^{L+1} \Mm(g_2)
\end{equation}
are compact operators from $\Kk_p$ to $\Kk_q$ for any $p,q \in[1,\infty]$.
\end{enumerate}
\end{prop}
\begin{proof}
 For (i) let us first assume that $h_1$ and $h_2$ are compactly supported and consider
 $\Mm(h_1)\Ff\Mm(h_2)$. Let $\Kk$ be the compact support of $h_2$.
 There exists a constant $C_\Kk$ such that for all $x\in \Kk$ and all $y \in\RR^k$ we have
 $|e^{ix\cdot y}-e^{ix\cdot y'}| \leq C_\Kk |y-y'|$.
 Therefore 
 \begin{align*}
|(\Ff \Mm(h_2) f)(y)-(\Ff \Mm(h_2)f)(y')| &\leq 
(2\pi)^{-k/2} C_K |y-y'|\,\| h_2 f\|_1 \\
& \leq (2\pi)^{-k/2} C_K \|h_2\|_2\,\|f\|_2\,|y-y'|\;
 \end{align*}
 and hence $\Mm(h_1) \Ff \Mm(h_2)$ maps a $L^2$ bounded sequence of functions into a sequence of equi-continuous functions,
 supported on the compact support of $h_1$. By the theorem of Arzela Ascoli we obtain a convergent subsequence in $L^\infty$ and
 hence in any $L^p$ norm.\\
 If $h_1$ and $h_2$ are continuous and exponentially decaying, then we can approach them in $\|\cdot\|_\infty$ norm
 by compactly supported continuous functions $h_{1,n},\, h_{2,n}$.
 Then $\Mm(h_{1,n}) \Ff \Mm(h_{2,n})$ approaches $\Mm(h_1) \Ff \Mm(h_2)$ in $L^2 \to L^p$ operator norm for any $p\in[1,\infty]$
(here, consider $\Mm(h_2)$ as map from $L^2$ to $L^1$, $\Ff$ as map from $L^1$ to $L^\infty$ and $\Mm(h_1)$ as map from $L^\infty$ to
$L^p$).
 
For the second operator in \eqref{op-base} note that by Remark~\ref{rem-FMF} for $D\in \Ss(k,L)$ Lemma~\ref{lem-FMF}~(iii) applies. By \eqref{eq-FMF3} one finds
after commuting the multiplication and shift operators on the left hand side that
$\Mm(h_1) \left(\Ff \Mm(e^{i x\cdot D x})\right)^{L+1} \Ff \Mm(h_2) =  \Ss_A \Mm(\hat h_1) \Ff \Mm(\hat h_2)$ where
$\hat h_1$ and $\hat h_2$ are exponentially decaying functions
and  $A$ is the product of the $D_j$ as in \ref{lem-FMF}~(iii).
Therefore, by the previous statement, this defines a compact operator from $L^2$ to any $L^p$, $p\in[1,\infty]$.

 Using \eqref{eq-T-parts} and the Leibniz-rule \eqref{eq-Leibn0} the statements (ii) and (iii)
 immediately follow from (i). 
 For the connection of the matrix $B$ in (ii) with $D$ as used in (i), note that the operator $T$ involves a Fourier transform on $\RR^{m\times 2n} \cong \RR^{2mn}$, so $k=2mn$,
 and combining the column vectors of $\bvp\in\RR^{m\times 2n}$ to one large vector $\varphi\in \RR^{2mn}$ vector,
 $\Tr(\bvp\bvp^\top B) = \Tr(\bvp^\top B \bvp)$ can be written as $\bvp\cdot D\bvp$ where $D$ is a $2mn\times 2mn$ matrix which is block-diagonal with the repeated $m\times m$ block $B$ along the diagonal, $D=\smat{B \\ & \cdot \\ & & B}$. Then $B\in\Ss(m,L)$ implies
 $D\in\Ss(2mn,L)$ so we can use part (i).
  Starting with an $\Hh$ bounded sequence
 one can subsequently construct a subsequence converging in all $L^2$ and $L^p$ norms involved in the definition of $\Hh_p$.
 Similar considerations can be made to obtain part (iii).
\end{proof}


\end{document}